\documentclass[10pt,onecolumn]{IEEEtran}
\pdfoutput=1
\usepackage{graphicx}
\usepackage{subfigure}
\usepackage{color}
\usepackage{textcomp}
\usepackage{amssymb}
\usepackage{citesort}
\usepackage{amsmath}
\usepackage{epsfig,graphics,epsf}
\usepackage{float,amsfonts}
\usepackage{fullpage}
\usepackage{citesort}
\newtheorem{theorem}{Theorem}[section]
\newtheorem{lemma}[theorem]{Lemma}
\newtheorem{corollary}[theorem]{Corollary}
\newtheorem{remark}{Remark}
\graphicspath{{figs//}}

\title{Capacity of All Nine Models of Channel Output Feedback for the Two-user Interference Channel}

\author{Achaleshwar Sahai, Vaneet Aggarwal, Melda Yuksel and
Ashutosh Sabharwal\footnote{A. Sahai and A. Sabharwal are with the
department of ECE, Rice University, Houston, TX 77005, USA (email:
\{as27,ashu\}@rice.edu). V. Aggarwal is with AT\&T Labs-Research, Florham Park, NJ 07932, USA (email: vaneet@research.att.com). M. Yuksel is with TOBB University of Economics and Technology, Ankara, Turkey (email:
yuksel@etu.edu.tr). The material in this paper was presented in part
at the IEEE Information Theory Workshop, Taormina, Italy, 2009 \cite{asahai1}
and at the IEEE International Symposium on Information Theory, Austin, Texas, 2010
\cite{asahai2}. A. Sahai and A. Sabharwal were partially supported by NSF grant CNS-1012921 and a grant from Texas Instruments.}}
\date{}
\begin{document}
\maketitle

\begin{abstract}

In this paper, we study the impact of different channel output feedback
architectures on the capacity of the two-user interference
channel. For a two-user interference channel, a feedback link can
exist between receivers and transmitters in 9 canonical architectures
(see Fig.~\ref{fig:allcombs}), ranging from only one feedback link to
four feedback links. We derive the exact capacity region for the
symmetric deterministic interference channel and the constant-gap
capacity region for the symmetric Gaussian interference channel for
all of the 9 architectures. We show that for a linear deterministic
symmetric interference channel, in the weak interference regime, all
models of feedback, except the one, which has only one of the receivers feeding back to
its own transmitter, have the identical capacity region. When only one of the receivers feeds back to its own transmitter, the capacity region is a strict subset of the capacity region of the rest of the feedback models in
the weak interference regime. However, the
sum-capacity of \emph{all} feedback models is identical in the weak
interference regime. Moreover, in the strong interference regime all
models of feedback with at least one of the receivers feeding back to
its own transmitter have the identical sum-capacity. For the Gaussian
interference channel, the results of the linear deterministic model
follow, where capacity is replaced with approximate capacity.
\end{abstract}

\section{Introduction}

The two-user interference channel has been studied in the literature
since 1970's to understand one of the main performance limits of
multiuser communication
networks~\cite{car,sato78,sato,costa82,costa85,sason,etkin}. Feedback
in interference channels has been considered in order to achieve a
possible improvement in data rates. A large body of work on
interference channels \cite{kramer02,kramer04,gast06,jiang07} explores
feedback strategies, where each receiver sends channel output feedback
to its own transmitter. More generalized form of feedback in a
two-user interference channel is considered in
\cite{tunita,tunout,tuninetti-sc,tuninetti}. Recent work in
\cite{suh2009,suh-allerton} particularly analyzes the capacity region
of two-user deterministic and Gaussian interference channels, where
each of the receivers send channel output feedback to its own
transmitter. The authors of \cite{alireza} consider the case of rate
limited channel output feedback and investigate its capacity region, where each user feeds back to its own
transmitter.

The conventional model of channel output feedback in a two-user
interference channel has each receiver feeding back to its intended
transmitter \cite{suh2009,alireza,suh-allerton}. However, several
different feedback architectures are possible based on the presence or
absence of feedback links between both receivers and both
transmitters. The feedback architecture can be asymmetric if feedback
resources available to different transmitter-receiver pairs are
different. Consider two mobile terminals in two neighboring cells,
communicating with their corresponding base stations. If the mobile
user in the first cell is closer to its base-station, then its
base-station can support a strong feedback link. At the same time, if
the mobile station in the neighboring cell is farther away from its
base-station, it will experience a poor or possibly no feedback
channel. In such a case, we say only one direct-link feedback is
available. In another scenario, suppose one of the receivers in the
interference channel is capable of sending feedback to both the
transmitters, whereas the other receiver does not send any feedback.
Then it would be a case of single receiver broadcasting feedback. The
conventional model of channel output feedback is \emph{insufficient}
to understand the effect of feedback on the capacity region of the
interference channel. We need to consider different feedback
architectures, which forms the focus of our study.

In this paper, we conduct a comprehensive study of the capacity region
of all feedback architectures in two-user linear deterministic
\cite{avestimehr,bresler2} and Gaussian interference channels. The
feedback architectures that we study are all parametrized by the
feedback links they support. In a two-user interference channel, there
can be as many as 4 possible feedback links, i.e., one feedback link
from each receiver to each transmitter. Therefore, excluding the case
of no feedback links, a total of $2^4 - 1 = 15$ feedback models are
possible. Barring the symmetrical cases, 9 canonical feedback models
are possible, which are shown in Fig.~\ref{fig:allcombs}. In this
work, we study the capacity region of all the 9 feedback models shown
in Fig.~\ref{fig:allcombs}. In order to gain insights about good
communication schemes that apply to the different feedback models, we
first analyze them under the symmetric linear deterministic model of
interference. Then, we extend the results to the Gaussian interference
channel, deriving the approximate capacity region by developing outer
and inner bounds, which are within constant bits of one another.

In this paper, the comprehensive study of capacity region of different
feedback architectures leads to three main results. The first main
result of the paper is that for a linear symmetric deterministic
interference channel, all 9 canonical feedback models except one (with
only one direct-link feedback, shown in Fig.~\ref{fig:fb_one-direct})
have the identical capacity region in the weak interference
regime. Moreover, the capacity region of single direct link feedback
model is a strict subset of the capacity region of the rest of the
feedback models. The first main result extends to the Gaussian channel
case where all models of feedback, except single direct-link feedback
model, have the same approximate capacity region which is within
constant bits from their respective outer-bounds.

The second main result of the paper is that for a linear symmetric
deterministic interference channel, all feedback models have the
\emph{identical} sum-capacity in the weak interference regime. This
result is particularly interesting because if sum-capacity is the
performance metric, any one feedback link is sufficient to achieve the
maximum feedback sum-capacity. 

The third main result of the paper is that to achieve maximum feedback
sum-capacity, availability of one direct feedback link is sufficient
for all regimes of interference, i.e., the sum-capacity with single
direct feedback link is identical to the sum-capacity with all four
feedback links for all regimes of interference. The second and third
main results also hold for the Gaussian interference channel, if the
term sum-capacity is replaced with approximate sum-capacity.

We show the above three results by deriving exact
(deterministic)/approximate (Gaussian) capacity regions of all of the
9 canonical feedback models. We find two new outer-bounds and propose
two new achievability schemes. For the deterministic channel model,
the achievability scheme attains all points on the outer bound,
whereas in the Gaussian model, the inner bound is a constant number of
bits away from the outer bound ($2.59$ bits/Hz for feedback models in
Fig.~2(a), 2(b), and 2(c), $4.59$~bits/Hz for Fig.~2(d) and
Fig.~2(e)). The achievability for all the feedback models is derived
in two steps. First, an achievable strategy is proposed for two atomic
feedback models: one with single direct feedback link and another with
single cross feedback link (where one of the receivers feeds back to
its interfering transmitter). Then, using a combination of the
achievable strategies for the two atomic feedback models, the
achievable rate region of the rest of the feedback models is derived.

 The first achievable strategy we propose for single direct feedback link
 is based on using a Han-Kobayashi type message
 splitting~\cite{han}. Our coding strategy is similar to the one employed in
 two-user interference channel without feedback in the sense that the
 coding scheme splits the message at each transmitter into two parts,
 private and common. However, the coding strategy differs in the
 transmission of the common message. The common message generated at
 the second transmitter is transmitted twice, once by the transmitter, where it is generated, and once from the other transmitter, where it is known via feedback. The purpose of the re-transmission of the common
 message depends on the regime of interference. In the strong
 interference regime, feedback offers gain, if it allows the common
 message to travel from its source to destination via an alternate
 independent path of higher capacity (than the direct link). In the weak interference regime, the first
 transmitter can perform block-Markov encoding based on the common
 message of the second transmitter. Block-Markov encoding of messages
 based on the common message of the second transmitter, helps the
 first receiver to resolve some of the past interference, without
 causing any apparent interference at the second receiver.

The above achievable strategy turns out to be insufficient to show the
exact/approximate capacity region for deterministic/Gaussian
interference channel with feedback models shown in
Fig.~\ref{fig:fb_others}. The second achievable strategy, for single
cross-link feedback model, is based on block-Markov encoding of
messages at the second transmitter and dirty paper encoding at the first
transmitter. Since the second transmitter performs block-Markov
encoding, and cross-link feedback is available to the first
transmitter, the first transmitter can learn about the ``future''
interference that its receiver will face. Based on the channel output
feedback from the second receiver, the first transmitter performs
dirty paper encoding to protect its receiver from future
interference. Using this second achievable strategy in combination with
the first achievable strategy, the capacity region for cross-link
feedback is proven.

{\bf Relations to similar work}: The coding strategy in
\cite{suh2009,tuninetti-sc,prabhakaran-sc} also employ a Han-Kobayashi
type message splitting. In \cite{suh2009}, the feedback model has each
transmitter receiving feedback from its respective receiver, and while
the message is split into only two parts, private and common, only a
part of the common message of the other transmitter is re-transmitted
in subsequent blocks. Our coding scheme for the single direct-link
feedback re-transmits all the common message of only one of the
transmitters. In \cite{tuninetti-sc,prabhakaran-sc}, the message is
split into four parts: two common and two private and feedback induces
source cooperation by making sources learn the common message of the
other transmitter. In our coding scheme too, the purpose of
re-transmitting the common message is to induce cooperation/allow
routing.

We would also like to remark that the work on generalized feedback in
\cite{tuninetti,tunout,tunita,tuninetti-sc}, as well as the work on
source cooperation by two sources overhearing each other's messages
over a noisy channel in \cite{prabhakaran-sc} are closely related to
our work. The outer and inner bounds derived in
\cite{tunita,tuninetti-sc} and \cite{prabhakaran-sc}, concurrent to
our work in \cite{asahai1,asahai2}, can be particularized to obtain
the sum-capacity result shown in Lemma~\ref{th:sumcap}. In this work,
we comprehensively study the exact and approximate capacity regions
for linear deterministic and Gaussian interference channel models
respectively for all canonical feedback models.

The rest of the paper is organized as follows. Section
\ref{SecPrelims} introduces the Gaussian channel model and its
deterministic approximation. Section~\ref{SecPrelims} also presents
all the different feedback models that will be studied in the
paper. Section~\ref{sec:mainresults} is a preview of the main results
and insights regarding them. Section~\ref{SecDC} and \ref{SecGC}
present the capacity regions (exact and approximate respectively) for
the linear deterministic and Gaussian interference channels for all
models of feedback. Section~\ref{sec:conclude} concludes the paper
with discussions.

\section{Channel Model and Preliminaries}
\label{SecPrelims}
In this section, we describe the two-user symmetric Gaussian and deterministic
interference channel models and the 9 canonical feedback architectures that will be used throughout the paper.

\subsection{Channel Model}
A two-user interference channel consists of two transmitters,
$\mathsf{T_1}$ and $\mathsf{T_2}$, and two receivers $\mathsf{D_1}$
and $\mathsf{D_2}$. Each receiver $\mathsf{D}_u$ is interested in the
message transmitted by transmitter $\mathsf{T}_u$ for $u\in\{1,2\}$,
while the message from the other transmitter is interference.

The two-user symmetric Gaussian interference channel, shown in
Fig.~\ref{fig:symmgauss} is a special case of the two-user interference
channel, where the noise at both the receivers have zero mean, unit
variance complex Gaussian distribution. Let $W_u$ denote the message
$\mathsf{T}_u$ transmits in $N$ successive transmissions,
where $W_u \in \mathcal{W}_u = \{1,2\ldots 2^{NR_u} \}$, $N \in
\mathbb{N}$ and $R_u \in \mathbb{R}$. The function $f_{uj}: W_u
\mapsto X_{uj}$ denotes the encoding that maps the message to the input over the
channel, $X_{uj}\in \mathbb{C}$, $j \in [1,2, \ldots
  N]$. Let $X_{u}^N = [X_{u1}, X_{u2}, \ldots X_{uN}]$ and $Y_{u}^N = [Y_{u1}, Y_{u2},
  \ldots Y_{uN}]$, where $X_{uj}$ ($Y_{uj}$) denotes the signal transmitted (received) at the $j^\mathrm{th}$ time instant at   $\mathsf{T}_u$ ($\mathsf{D}_u$). Then, when $g_{ij}\in \mathbb{C}$ are the channel gains, the received signals at the two receivers are given by
\begin{eqnarray}
 Y_{1j} & = & g_{11}X_{1j} + g_{21}X_{2j} + Z_{1j} \nonumber \\
 Y_{2j} & = & g_{22}X_{2j} + g_{12}X_{1j} + Z_{2j} \nonumber .
\end{eqnarray}
The decoding function $h_u$ maps the output $Y_{u}^N$ to a symbol
$\widehat{W}_u \in \mathcal{W}_u$ ($h_u: Y_{u}^N \mapsto
\widehat{W}_u$).

In this paper, we will focus on the symmetric Gaussian channel, where the direct gains
are equal, $g_{11} = g_{22} = g_{d}$, the cross gains are equal,
$g_{12} = g_{21} = g_c$, and the noises $Z_{1j}$ and $Z_{2j}$ are both
distributed as $\mathcal{CN}(0,1)$. Moreover, the transmitted power is
constrained such that $\mathbb{E}(|X_{1j}|^2) \le P_1$,
$\mathbb{E}(|X_{2j}|^2) \le P_2$, and $P_1 = P_2 = P$, where the
$\mathbb{E}(.)$ denotes the expected value of a random variable. We
also define the signal to noise ratio ($\mathsf{SNR}$) and the
interference to noise ratio ($\mathsf{INR}$) as
\begin{equation}
 \mathsf{SNR} = |g_d|^2P,\text{ } \mathsf{INR} = |g_c|^2P. \nonumber
\end{equation}
The regime of interference is weak, when $\mathsf{SNR} \geq
\mathsf{INR}$ and strong when $\mathsf{SNR} < \mathsf{INR}$.
Moreover, the ratio of $\mathsf{INR}$ to $\mathsf{SNR}$ in dB scale
will be denoted by
\begin{equation}
\alpha = \frac{\log(\mathsf{INR})}{\log(\mathsf{SNR})}.
\end{equation}

\begin{figure}
\centering \subfigure[]{\label{fig:symmdeter}
\resizebox{1.8in}{!}{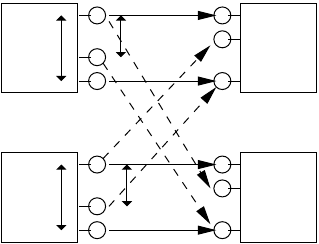}}
\hspace{0.5in}
\subfigure[]{\label{fig:symmgauss}
\resizebox{2.2in}{!}{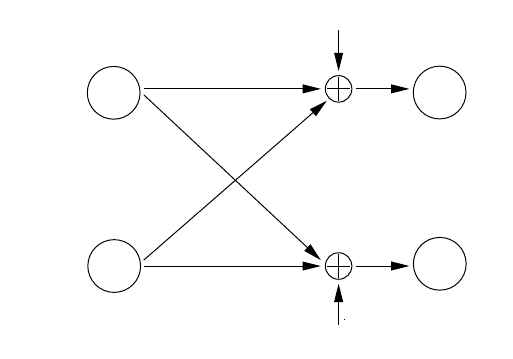}}
   \caption {The (a) deterministic and (b) Gaussian models for the two-user interference channel.}
 \label{fig:ifc}
\end{figure}

The deterministic interference channel \cite{avestimehr} is a good
approximation of the Gaussian interference channel, when signal and
interference powers are much larger compared to the noise. We will use
the deterministic approximation of the two-user Gaussian interference
channel with feedback to develop insights for designing achievable communication
strategies for the Gaussian model. The
deterministic interference channel is described as follows. Associated
with the link between transmitter $\mathsf{T}_u$, $u \in \{1,2\}$, and
receiver $\mathsf{D}_{k}$, $k \in \{1,2\}$, is a non-negative integer
$n_{uk}$ (which corresponds to the channel gain in the Gaussian
channel). Let $q = \max_{u,k} (n_{uk})$. Overloading the notation for
input and output, the inputs at $u^{\mathrm{th}}$ transmitter at time
$j$ is denoted by $X_{uj} \in \mathbb{F}_2^q$. Equivalently, $X_{uj}$
can be written as $X_{uj} =\left[ X_{uj_1}X_{uj_2} \ldots
  X_{uj_q}\right]^T $, such that $X_{uj_1}$ and $X_{uj_q}$ are the
most and the least significant bits respectively. The received signal
at time $j$ is denoted by the vector $Y_{kj}\in \mathbb{F}_2^q$ or
equivalently $Y_{kj} =\left[ Y_{kj_1}Y_{kj_2} \ldots Y_{kj_q}\right]
^{T}$. Specifically, the received signal $Y_{kj}$, $k=1,2,$ of a deterministic
interference channel is given by
\begin{equation}%
\begin{array}
[c]{cc}%
Y_{kj}=\mathbf{S}^{q-n_{1k}}X_{1j}\oplus\mathbf{S}^{q-n_{2k}}X_{2j} &
k=\{1,2\},
\end{array}
\end{equation}
where $\oplus$ denotes the XOR operation, and $\mathbf{S}$ is a
$q\times q$ shift matrix with ones on the first diagonal below the
main diagonal, and zeros everywhere else. The symmetric deterministic
channel, shown in Fig.~\ref{fig:symmdeter}, is characterized by two
values: $n = n_{11} = n_{22}$ and $m = n_{12} = n_{21}$. Here $n$ and
$m$ indicate the number of signal bit levels that we can send through
the direct links and the cross links, respectively. When $\frac{m}{n}
\leq 1$, the system is in the weak interference regime, and when
$\frac{m}{n} > 1$, the system is in the strong interference regime. We
denote by $\mathbf{O}_{p} = [0,0, \ldots, 0]^T $ such that the
cardinality of $\mathbf{O}_{p}$ is $p$.


\subsection{Feedback Models}

In this paper, we will use feedback to imply channel output feedback
from the receivers to the transmitters. The feedback is assumed to be
strictly causal and noiseless. There are four feedback links from the
two receivers to the two transmitters. A feedback model is defined by
the four-tuple $(F_{11} F_{12} F_{21} F_{22})$, where
\begin{equation}
F_{ku} = \begin{cases} 1 & \text{if there is a feedback link from }
  \mathsf{D}_k \text{ to } \mathsf{T}_u, \\ 0 & \text{otherwise}.
\end{cases}
\end{equation}
Fig.~\ref{fig:allcombs} shows the 9 principal feedback combinations and lists their
symmetrical equivalent feedback models. With
feedback, we can formalize the transmitted symbols as
\begin{equation}
  X_{uj} = f_{uj}(W_u, Y_1^{j-1}F_{1u}, Y_2^{j-1} F_{2u}),\text{ }u=
  \{1,2\}, \nonumber
\end{equation}
where $F_{ku} = 1$ implies that the channel output, $Y_{k}^{j-1}$, is
known causally to the $u^{\mathrm{th}}$ transmitter. The feedback link
from a receiver to its own transmitter is the direct-link feedback,
and the link to the other transmitter is the cross-link feedback. If
only two direct-links of feedback exist, then $F_{11}=F_{22}=1$ and
$F_{12}=F_{21}=0$. If only one direct-link feedback exists, then
$F_{12}=F_{21}=0$ and either $F_{11}$ or $F_{22}$ is 1 while the other
is zero. Since we consider the symmetric interference channel, unless
otherwise specified, we will, without loss of generality, assume that
one direct-link feedback model is equivalent to $F_{11}=1$ and
$F_{22}=0$.
When feedback is broadcast from a single receiver, we will assume that
$F_{11}=F_{12}=1$ while $F_{21}=F_{22}=0$.
In the $(1111)$ feedback model, both receivers broadcast their channel outputs.

\subsection{Achievable Rate and Capacity Definitions}
A rate pair $(R_1, R_2)$ is said to be \emph{achievable}, if for
independent and identically distributed (i.i.d.) messages $W_1 \in \mathcal{W}_1$ and $W_2 \in
\mathcal{W}_2$, where $\mathcal{W}_u = \{1, \cdots, 2^{NR_u}\}$ and $u
\in \{1,2\}$, there exist encoders $f_{uj}$ and decoders $h_u$ so that
the probability that the decoded messages $\widehat{W}_1$ and
$\widehat{W}_2$ at ${\sf D}_1$ and ${\sf D_2}$ respectively are in
error goes to 0 as $N\to\infty$. More precisely, for $u = \{1,2\}$ define the average
error probability of the message $\mathsf{T}_u$ transmits to $\mathsf{D}_u$
as
\begin{equation}
\epsilon_{u,N} = \mathbb{E}(\mathrm{Pr}(\widehat{W}_u \neq W_u )).
\end{equation}
Then the rate pair $(R_1, R_2)$ is said to be achievable, if both
$\epsilon_{1,N}$ and $\epsilon_{2,N}$ can be driven to zero as $N \to
\infty$.  The capacity region is the closure of all achievable rate
pairs $(R_1,R_2)$. Since there are different capacity regions for
different feedback models, we will use a superscript representing the
state $(F_{11} F_{12}F_{21}F_{22})$. The capacity region and the
sum-capacity of the $(F_{11} F_{12}F_{21}F_{22})$ feedback model are
respectively denoted by $\mathcal{C}^{(F_{11} F_{12}F_{21}F_{22})}$
and $C_\mathrm{sum}^{(F_{11}F_{12}F_{21}F_{22})}$, while the
achievable rate region and the sum-rate are denoted by $\mathcal{R}^{(F_{11} F_{12} F_{21} F_{22})}$ and $R^{(F_{11} F_{12}
  F_{21} F_{22})}_{\rm sum}$ respectively.


\begin{figure}[t]
\begin{center}
\subfigure[Feedback models with at least two direct feedback
  links]{\label{fig:both-direct}
  \resizebox{!}{1.0in}{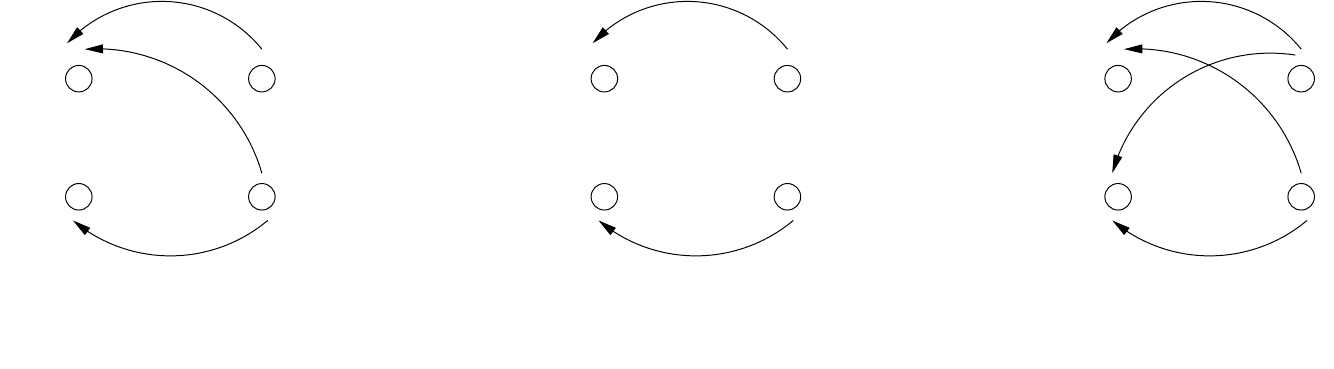}}
\\ \subfigure[One of the receivers is broad-casting
  feedback]{\label{fig:fb_one-broadcast}
  \resizebox{!}{1.0in}{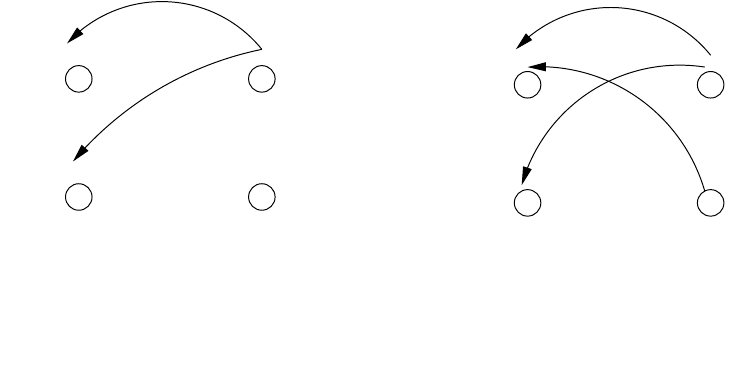}} \hspace{0.2cm}
\subfigure[Both receivers feeding back to the unintended
  transmitter]{\label{fig:fb_both-cross}
  \resizebox{!}{1.0in}{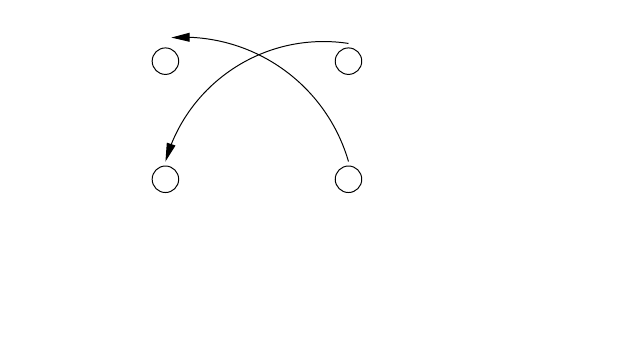}} \hspace{0.2cm}
\subfigure[Only one receiver feeding back to its
  transmitter]{\label{fig:fb_one-direct}
  \resizebox{!}{1.0in}{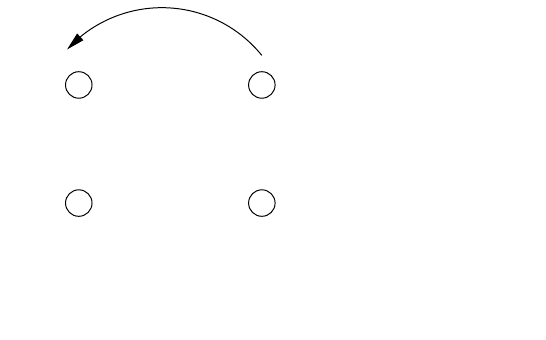}}
\\ \subfigure[Feedback models where only one of the transmitter
  receives cross-link or cross- as well as direct-link feedback]{\label{fig:fb_others}
  \resizebox{!}{1.0in}{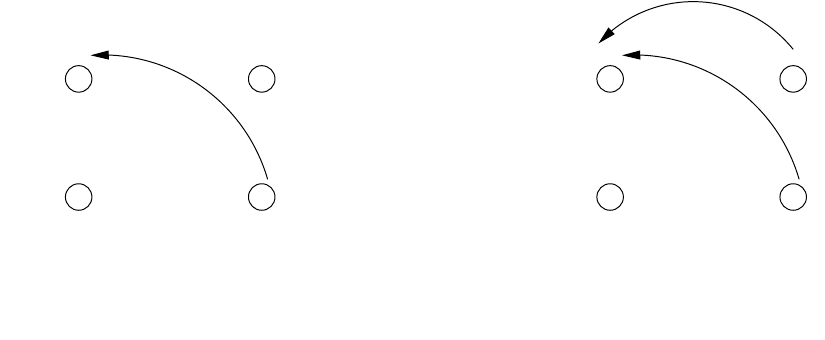}}
\caption{The 9 canonical feedback models. The figure
  shows only the feedback links, while the underlying interference
  channel is depicted in Fig.~\ref{fig:ifc}. The feedback state of
  each of the feedback models is also shown. These 9 models (15
  including the symmetric cases) constitute all possible cases of
  feedback.}\label{fig:allcombs}
\end{center}
\end{figure}

\subsection{Prior Results}
To contrast the capacity region and the sum-capacity results derived
in this paper to the no feedback case, the following theorem is presented:
\begin{theorem}[\cite{costa82,bresler2,etkin}] \label{lem:det_region}The
  capacity region of the two-user symmetric deterministic interference
  channel without any feedback, $\mathcal{C}^{(0000)}$ is the closure
  of all $(R_1,R_2)$ satisfying
\begin{eqnarray}
R_1&\le&n \label{eq:th1r1} \\ R_2&\le&n \label{eq:th1r2} \\ R_1+R_2&\le&
\min((n-m)^+ + \max(m,n), 2 \max(m , (n -m))) \label{eq:th1r1r2}  \\
R_1+2R_2&\le& \max(m,n) + (n - m)^+ + \max(m, (n - m)) \label{eq:th1r12r2}  \\
2R_1+R_2&\le& \max(m,n) + (n - m)^+ + \max(m, (n - m)). \label{eq:th12r1r2}
\end{eqnarray}
\end{theorem}

The capacity region of the two-user deterministic interference
channel, with feedback from both the receivers to their respective
transmitters, i.e. the $(1001)$ feedback model, has been studied in
\cite{suh-allerton,suh2009} and is given by:
\begin{theorem}[\cite{suh-allerton}]\label{th:psumcap} \label{lem:det_region_fb}
 The capacity region of the two-user symmetric deterministic
 interference channel with two direct feedback links,
 $\mathcal{C}^{(1001)}$, is the closure of all $(R_1,R_2)$ satisfying
\begin{eqnarray}
R_1&\le& \max(n,m) \label{eq:th2r1} \\ R_2&\le& \max(n,m) \label{eq:th2r1}
\\ R_1+R_2&\le& (n - m)^+ + \max(n, m). \label{eq:th2r1r2}
\end{eqnarray}
\end{theorem}
Theorem \ref{th:psumcap} shows that with feedback, the capacity region
of the deterministic interference channel enlarges and the
sum-capacity improves, as the $(1001)$ feedback model deactivates the
bounds \eqref{eq:th1r12r2} and \eqref{eq:th12r1r2} in Theorem
\ref{lem:det_region}.


\section{Preview of Main Results}
\label{sec:mainresults}

In this paper, we will prove the capacity/approximate-capacity region
of all 9 canonical feedback models for the deterministic/Gaussian
channels. Before presenting the technical details in Sections
\ref{SecDC} and \ref{SecGC}, in this section, we highlight our main
results and offer related insights.
\begin{enumerate}
\item Except the single direct-link feedback model $(1000)$, all
  feedback models have the identical capacity region in the weak
  interference regime. Thus, all feedback models (except the $(1000)$
  feedback model) achieve the capacity region achievable by all four
  feedback links, $\mathcal{C}^{(1111)}$. In particular, this result includes that the capacity
  region of the single cross-link feedback model is identical with the capacity
  region with all four feedback links, i.e., $\mathcal{C}^{(0010)}
  \equiv \mathcal{C}^{(1111)}$ in the weak interference regime. Moreover, the capacity region
  $\mathcal{C}^{(1000)}$ is a strict subset of
  $\mathcal{C}^{(1111)}$. 
\item All feedback models have the identical sum-capacity in the weak
  interference regime.
\item In the strong interference regime, feedback models with at least
  one direct-link of feedback have the same sum-capacity as that with
  all four feedback links, i.e., $C_{\rm sum}^{(1000)} = C_{\rm
    sum}^{(1\mathsf{xxx})} = C_{\rm sum}^{(1111)}$.
\end{enumerate}
The above results that hold for deterministic channels apply to
Gaussian channels, if the term capacity is replaced with approximate
capacity. We develop the above results by deriving two new
outer-bounds and proposing two new achievability schemes. An
illustration of the achievability schemes through examples and
intuitions about the above results follow.

\subsection{Weak Interference Regime} \text{ }
{\bf Gain due to source cooperation}: If a source receives feedback,
it can \emph{causally} learn a part of the message being transmitted
by the other source. Thus source cooperation can be induced, which
improves the capacity region and the sum-capacity. For instance, let ${\sf
  T}_1$ receive feedback, then it can causally learn a part of the
message transmitted by the interfering source, ${\sf T}_2$. If the
``past'' interference impairs decoding the intended signal at ${\sf
  D}_1$, then with the help of causal feedback, ${\sf T}_1$ can encode
its message in the forthcoming blocks such that it can help its
receiver resolve the ``past'' interference. On the other hand, the
knowledge of the message transmitted by ${\sf T}_2$ can also be used
by ${\sf T}_1$ to encode its message such that it is robust against
``future'' interference from ${\sf T}_2$. We illustrate the two forms
of source cooperation, which are possible in direct-link and cross-link
feedback models through two examples in a deterministic model with
$\frac{m}{n} = \frac{1}{3}$.

\emph{Example 1, direct-link feedback:} In the coding strategy shown
in Fig.~\ref{fig:relay_31}, $\mathsf{T}_1$ learns the interference,
$b_1$, received at $\mathsf{D}_1$ in the first block via feedback and
transmits it in the second block on a bit-level that is above its
interference floor. This enables $\mathsf{D}_1$ to decode the
interference that occurred in the first block, while causing no
apparent interference at ${\sf D_2}$ (since $b_1$ is an intended
signal at ${\sf D}_2$). With the number of blocks approaching $\infty$, the
rate pair $(2,3)$ is achievable.

\emph{Example 2, cross-link feedback:} In the coding strategy shown in
Fig.~\ref{fig:dirty_31}, ${\sf T}_1$ learns the message transmitted by
${\sf T}_2$ in the first block. The transmitter ${\sf T}_2$ follows a
block-Markov type encoding and repeats $b_2$ in the second time
block. Since ${\sf T}_1$ knows $b_2$ at the end of first block, via
cross-link feedback, and is also aware that $b_2$ is the likely
interference in the second block, it performs a dirty paper like
encoding scheme to ensure that its message is robust to interference. In the
second block three bits of intended message are decodable at ${\sf
  D}_1$ and two bits are decodable at ${\sf D}_2$, thus leading to
rate pair $(3,2)$ as number of blocks $\to \infty$.

In either of the examples above, one of the transmitter-receiver pairs
communicates essentially ``interference-free'', even though the other
transmitter is transmitting at bit-levels which cause
interference. Such interference-free communication is impossible
without feedback, unless one of the transmitter-receiver pair
sacrifices its rate. An important difference between direct and
cross-link feedback is that with direct-link feedback, ${\sf T}_1$ can
know only the ``past'' interference, while with cross-link feedback
$\mathsf{T}_1$ has access to possibly both ``past'' and ``future''
interference. Thus, with cross-link feedback, the rate pairs $(3,2)$
as well as $(2,3)$ are achievable if $\frac{m}{n} =
\frac{1}{3}$. However, with single direct-link feedback, where ${\sf
  T}_1$ is the only source receiving feedback, the rate pair $(3,2)$
is \emph{not} achievable. Therefore, in the weak interference regime,
cross-link feedback model has a larger capacity region than
direct-link link feedback, $\mathcal{C}^{(1000)} \subset
\mathcal{C}^{(0010)}$. Also, if both direct feedback links are
present, then by symmetry both rate pairs $(3,2)$ as well as $(2,3)$
are achievable. Thus, $\mathcal{C}^{(1000)} \subset
\mathcal{C}^{(1001)}$.

\begin{figure}[t]
\centering \subfigure[Using the direct-link $(1000)$ feedback model, $b_1$
  received in the second block at $\mathsf{D}_1$ helps resolve the
  interference at ${\sf D}_1$ in the first block. Also note that all the
  intended bits can be decoded at $\mathsf{D}_2$, and thus there is no
  interference observed at ${\sf D}_2$. The rate pair $(2,3)$ is achievable,
  when number of blocks $\to
  \infty$.]{\label{fig:relay_31}\scalebox{0.4}{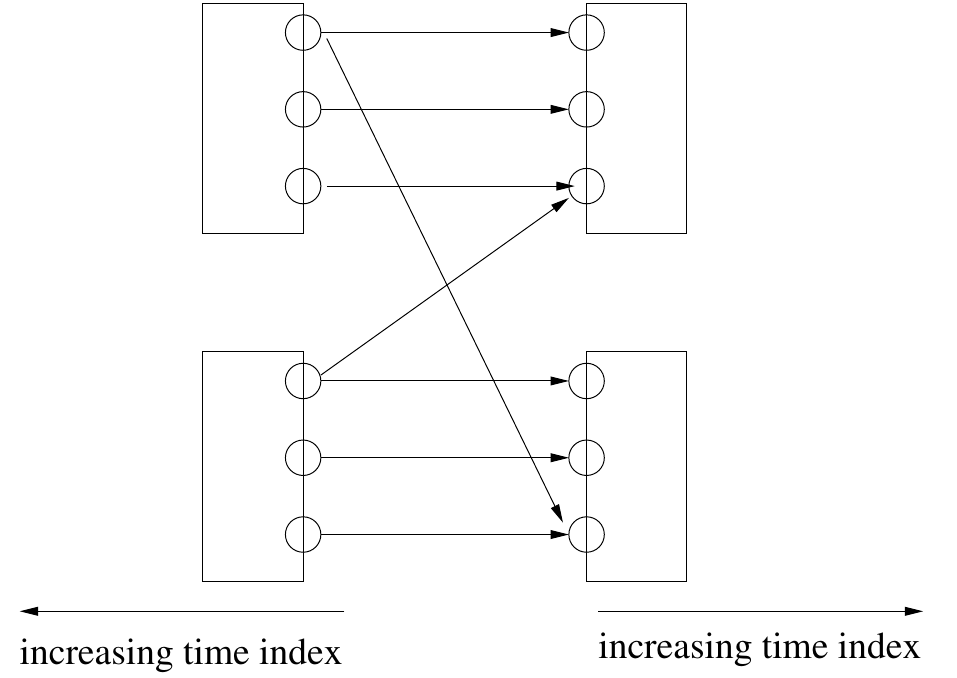}}
\hspace{1cm} \subfigure[Using the cross-link $(0010)$ feedback model,
  $b_2$ is known at ${\sf T}_1$ at the end of the first block of
  transmission. Using the knowledge of $b_2$, dirty paper encoding is
  performed at ${\sf T}_1$, such that the interference from ${\sf
    T}_2$ does not impair reception at ${\sf D}_1$. The rate
  pair $(3,2)$ is achievable, when number of blocks $\to
  \infty$.]{\scalebox{0.4}{\label{fig:dirty_31}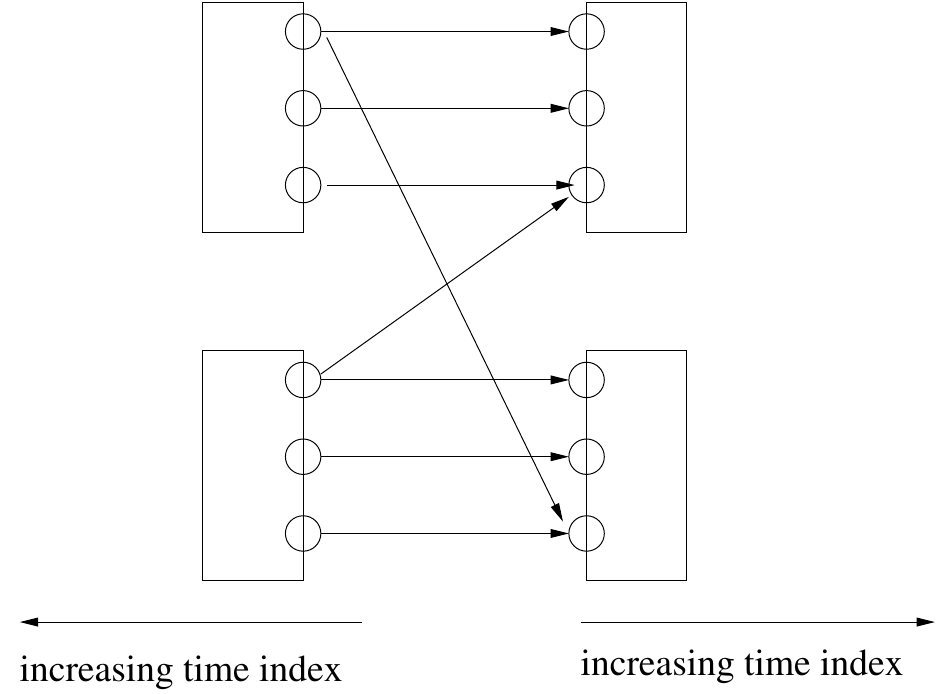}}
\caption{The first two blocks of encoding for $(1000)$ and $(0010)$ feedback models for the deterministic interference channel with $n = 3, m =1$.}
\end{figure}

{\bf Limited gain due to feedback delay}: Feedback implies that
cooperation between sources can occur only causally. In the example
shown in Fig~\ref{fig:relay_31}, ${\sf T}_1$ expends resources (bits)
to help its receiver resolve ``past'' interference, while in the example
shown in Fig~\ref{fig:dirty_31}, ${\sf T}_2$ expends resources in
creating known interference at ${\sf D}_1$. Even with all
four feedback links present, i.e., $(1111)$ feedback model, there is a
trade-off between expending resources to transmit a new message
versus resolving past interference/creating known
interference. Therefore, in the weak interference
regime, having all four feedback links does not enlarge the capacity
region compared to two direct-link feedback or cross-link feedback,
i.e., $\mathcal{C}^{(0010)} \equiv \mathcal{C}^{(1001)} \equiv
\mathcal{C}^{(1111)}$.

{\bf Equivalence of sum-capacity}: The capacity region of the single
direct-link feedback model is smaller than the rest of the feedback
models. This is so because in single direct-link feedback model,
unlike other feedback models, only one of the sources, say ${\sf
  T}_1$, can assist the other source, ${\sf T}_2$ such that there is
no apparent interference at its intended receiver ${\sf
  D}_2$. However, such one-sided assistance is sufficient to achieve
the same sum-capacity as would be achievable with two sided assistance
(possible with cross-link, two direct-link or all four feedback
links). Thus, ${C}_{\rm sum}^{(1000)} = {C}_{\rm
  sum}^{(1\mathsf{xxx})} = C_{\rm sum}^{(0010)} = C_{\rm
  sum}^{(1111)}$.




\subsection{Strong Interference Regime} \text{ }
In the strong interference regime, feedback offers improvement in both
the sum rate and the capacity region, if it enables an alternate
independent path of higher capacity for messages to travel from its
source to its destination. As a direct consequence, any feedback model,
which does not lead to an alternate path, attains no gain. On the
other hand, in feedback models, which obtain gains out of feedback
(models with at least one direct feedback link), in the strong
interference regime, the gain is limited by the capacity of the
alternate path. We describe how this limitation leads to the result
that all feedback models with at least one direct feedback link have
the same sum-capacity.

\begin{figure}[t]
\begin{center}
\subfigure[The dashed line indicates the alternate path from
  $\mathsf{T}_2$ to
  $\mathsf{D}_2$.]{\label{fig:altpath}\scalebox{0.4}{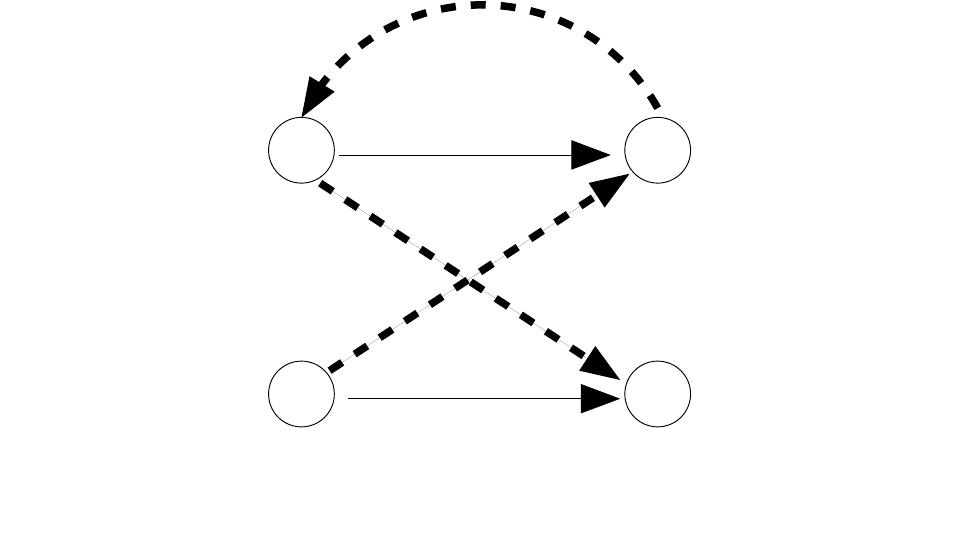}} \hspace{2cm}
\subfigure[Both alternate paths for communicating a message from a
  transmitter to its intended receiver are depicted, one is the dashed
  line and the other is the dotted line. The alternate paths share a
  common, finite capacity
  sub-path.]{\label{fig:altpath2}\scalebox{0.4}{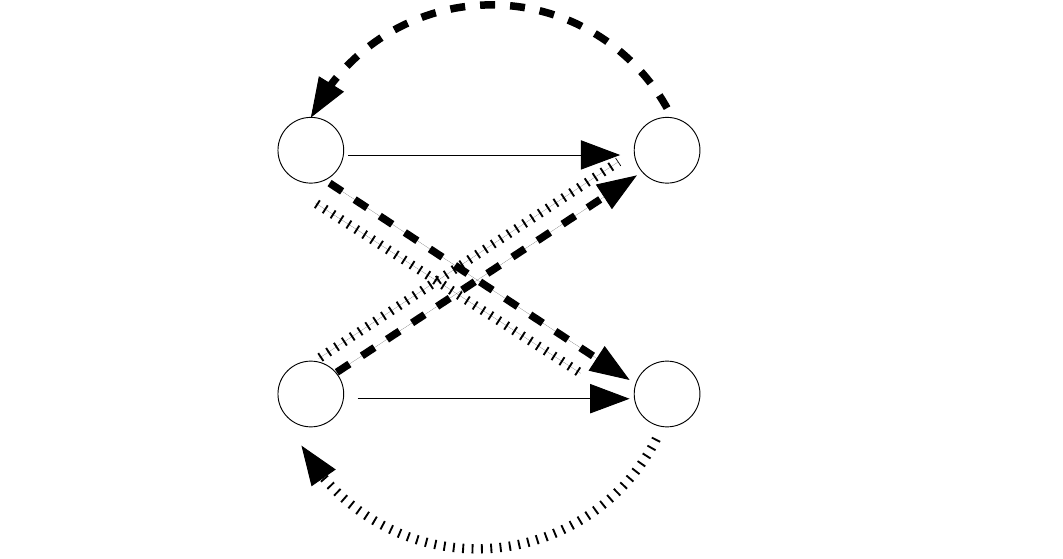}}
\caption{Alternate paths, which improve rates in the strong
  interference regime.}
\end{center}
\end{figure}

{\bf Gain due to alternate independent path}: In
Fig.~\ref{fig:altpath}, single direct-link feedback enables an
alternate independent path for messages to travel from $\mathsf{T}_2$
to $\mathsf{D}_2$. The feedback link between $\mathsf{D}_1$ and
$\mathsf{T}_1$, in conjunction with the interfering links between
$\mathsf{T}_2$-$\mathsf{D}_1$ and $\mathsf{T}_1$-$\mathsf{D}_2$ forms
the alternate path. Since the interfering links are stronger than the
direct link and feedback is of infinite capacity, the rate at which
$\mathsf{T}_2$-$\mathsf{D}_2$ can communicate is higher than the rate
possible without feedback. Note that, only
$\mathsf{T}_2$-$\mathsf{D}_2$ pair has an alternate independent path,
and therefore only the rate of $\mathsf{T}_2$-$\mathsf{D}_2$
increases. On the other hand, in the example shown in
Fig.~\ref{fig:altpath2}, with two direct feedback links, both
transmitter-receiver pairs have alternate independent
paths. Consequently rates of both source-destination pairs can be
boosted. Therefore, the capacity region achievable with the feedback model
with both direct feedback links is larger than the capacity region
achievable with the feedback model with only one direct feedback link,
i.e., $\mathcal{C}^{(1000)} \subset \mathcal{C}^{(1001)}$.

A key commonality in the feedback models shown in
Fig.~\ref{fig:altpath} and Fig.~\ref{fig:altpath2} is that the
alternate independent path in both feedback models necessarily
contains the pair of interfering links, ${\sf T}_1$-${\sf D}_2$ and
${\sf T}_2$-${\sf D}_1$ as a resource that is intelligently used to
boost the rate. The increase in the sum rate is limited by the capacity of
the shared resource, i.e.~the capacity of the interference links. Thus,
whether there is single direct feedback link or two direct feedback
links, the same gain in the sum-rate is possible, thus the sum-capacity of
all feedback models with at least one direct feedback link are
identical, i.e., $C_{\rm sum}^{(1000)} = C_{\rm sum}^{(1\mathsf{xxx})}
= C_{\rm sum}^{(1111)}$.

{\bf Cross link feedback creates no alternate path}: Any model of
feedback, which has only cross feedback links, does not result in any
alternate independent path for messages to travel from the source to
its destination. Consequently, no improvement in the individual rate
or in the rate region is observed. Thus, the capacity region with or
without cross link feedback are the same in the strong interference
regime, i.e. $\mathcal{C}^{(0110)} \equiv
\mathcal{C}^{(0000)}$. As the cross links do not bring in any gains, in the strong interference regime, the capacity
region of the feedback model with all four feedback links is the same
as the capacity region of the feedback model with only two direct
feedback links, i.e. $\mathcal{C}^{(1\mathsf{xx}1)} \equiv
\mathcal{C}^{(1111)}$.

\section{Feedback models: Deterministic channels}\label{SecDC}

In this section, we first present Theorem~\ref{th:detcapreg} and
Theorem~\ref{th:detsumcap}, which respectively state the capacity region and the
sum-capacity of all 9 canonical feedback models for the linear
deterministic channel. In Section~\ref{subsec:obound}, we provide
outer-bounds on the capacity region and the sum-capacity in
Lemma~\ref{lem:cutset} \cite{prabhakaran-sc}, Lemma~\ref{th:sumcap}
\cite{asahai1} and Lemma~\ref{th:sumcap2}. In
Sections~\ref{subsec:atomic}, \ref{subsec:capreg} and
\ref{subsec:sumcap}, we show the achievability of the capacity region
and the sum-capacity of all 9 feedback models.

\begin{theorem}\label{th:detcapreg}
  The capacity regions of the two-user symmetric deterministic
  interference channel for all the 9 canonical feedback models are
  given in Table~\ref{table:det_capreg}.
  \begin{table}[ht]
    \centering
    \caption{Capacity regions for all 9 canonical feedback models}
    \begin{tabular}{|c|l|}
      \hline & \\
      Feedback Models & { Capacity Region} \\
      & \\ \hline & \\
      & $R_1 \leq \max(n,m)$ \\
      $(1\mathsf{x} \mathsf{x}1)$ &  $R_2 \leq \max(n,m)$ \\
       & $R_1 + R_2 \leq (n - m)^+ + \max(n,m)$ \\ & \\ \hline & \\
      & $R_1 \leq n $ \\
      $(1100), (1110)$ & $R_2 \leq \max(n,m)$ \\
      $(1010)$ & $R_1 + R_2 \leq (n-m)^+ +\max(n,m)$ \\ & \\ \hline & \\
       & $R_1 \leq n$ \\
      $(0110)$, $(0010)$ & $R_2 \leq n$ \\
      & $R_1 + R_2 \leq (n-m)^+ + \max(n,m)$ \\ &  \\ \hline & \\
      & $R_1 \leq n$ \\
      $(1000)$ & $R_2 \leq \max(n,m)$\\
      & $R_1 + R_2 \leq  (n-m)^+ +\max(n,m)$\\
      & $2R_1 + R_2 \leq (n-m)^+ + \max(n,m) + \max(n - m , m)$ \\ & \\
      \hline
    \end{tabular}\label{table:det_capreg}
  \end{table}
\end{theorem}

\begin{theorem}\label{th:detsumcap}
  The sum-capacity of the two-user symmetric deterministic interference channel for all the 9 canonical feedback models is given in Table~\ref{table:det_sumcap}.
  \begin{table}[ht]
  \caption{Deterministic sum-capacity for all 9 canonical feedback
    models} \centering
\begin{tabular}{|c|c|}
\hline & \\ Feedback Models & Sum-capacity \\ & \\ \hline &
\\ $(1\mathsf{xxx})$ & $ (n-m)^+ + \max(n,m)$ \\ & \\ \hline &
\\ $(0110), (0010)$ & $\min\{(n-m)^+ + \max(n,m), 2n\}$ \\ & \\ \hline
\end{tabular} \label{table:det_sumcap}
\end{table}
\end{theorem}

\subsection{Outer Bounds}\label{subsec:obound}

Feedback in interference channels is a special case of source
cooperation. Thus the cut-set bounds on the interference channel with
source cooperation apply to interference channels with feedback as
well. In this subsection, along with the cut-set bound for
interference channels with feedback, we describe two new outer-bounds
for different feedback models.

\begin{lemma} [\cite{coverbook,prabhakaran-sc}] \label{lem:cutset}
  The cut-set and no-interference bound for different feedback combinations is given by
\begin{eqnarray}
\label{cutset1} R_1 & \leq & \max(n,c_1)  \\
\label{cutset2} R_2 & \leq & \max(n,c_2),
\end{eqnarray}
where
\begin{equation}
     c_1 = \begin{cases}
       0 & \text{if ${\sf T}_2$ receives no direct-link feedback} \\
       m & \text{otherwise}
     \end{cases}
\end{equation}
and
\begin{equation}
     c_2 = \begin{cases}
       0 & \text{if ${\sf T}_1$ receives no direct-link feedback} \\
       m & \text{otherwise}
     \end{cases}.
\end{equation}
\end{lemma}

Next, we present an outer-bound on the sum-capacity of the feedback
model $(1111)$ we derived in~\cite{asahai1}. Concurrent to
\cite{asahai1}, the authors in \cite{tunout} derive an outer-bound for
the generalized feedback model. The authors in \cite{prabhakaran-sc} also derive
outer bounds for interference channels with source cooperation, which
can be particularized for the linear symmetric deterministic
interference channel to obtain the same result.

\begin{lemma} [\cite{tunout,prabhakaran-sc,asahai1}]\label{th:sumcap}  The sum-capacity of the feedback model $(1111)$, $\mathcal{C}_{\rm sum}^{(1111)}$, is outer bounded by
\begin{equation}\label{boundsumcap}
R_1 + R_2 \leq (n-m)^+ + \max(n,m).
\end{equation}
\end{lemma}
\begin{remark} Since none of the feedback models can have a sum-capacity
 larger than the sum-capacity for the $(1111)$ feedback model, (\ref{boundsumcap}) is an
 outer-bound on the sum-capacity of all feedback models.
\end{remark}

\begin{lemma} \label{th:sumcap2} The capacity region of the two-user symmetric deterministic interference channel with feedback state $(1000)$ is outer bounded by
\begin{equation} \label{eqsumcap2}
2R_1 + R_2 \leq (n -m)^+ + \max(n,m) + \max(m, n - m).
\end{equation}
\end{lemma}
\begin{proof}
The proof is provided in Appendix \ref{apd:2}.
\end{proof}
\begin{remark}
We note that \eqref{eqsumcap2} is identical to \eqref{eq:th12r1r2}, i.e., the bound on $2R_1 + R_2$, when there is
no feedback.
\end{remark}

\subsection{Achievability for Two Atomic Feedback Models}
\label{subsec:atomic}
To show the achievability of the capacity region of all 9 canonical
cases of feedback, we first show an achievable strategy for the single
direct-link feedback model, which is based on Han-Kobayashi type
message splitting~\cite{han}. Then, to show the achievability of the
single cross-link feedback model, we present Lemma~\ref{ThCrslnkD},
which allows us to connect the achievable rate region of the single
direct-link feedback model with the achievable rate-region of the
single cross-link feedback model. To complete the achievability of the
single cross-link feedback model, we show another achievable strategy,
which is based on block-Markov and dirty paper encoding and
decoding. Finally, using the achievability for the single direct-link
feedback and single cross-link feedback models, we show the
achievability of the capacity region of all 9 canonical feedback
models.

To show the achievability of the capacity region of single direct-link
and single cross-link feedback models, we establish the achievability
of the corner points formed by the intersection of the outer-bounds given
by \eqref{cutset1}, \eqref{cutset2}, \eqref{boundsumcap} and
\eqref{eqsumcap2}. Since the capacity regions are formed by the
intersection of hyper-planes, the capacity regions are
convex polygons. The achievability of a convex polygon is proved, if
the non-trivial corner points of the convex polygon are shown to be
achievable. We define the following points
\begin{eqnarray}
  \mathcal{K}_{\rm A} = \text{\{$(R_1,R_2)$: \eqref{cutset1} and \eqref{eqsumcap2} hold with equality simultaneously\}},\nonumber \\
  \mathcal{K}_{\rm B} = \text{\{$(R_1,R_2)$: \eqref{cutset1} and \eqref{boundsumcap} hold with equality simultaneously\}}, \nonumber \\
  \mathcal{K}_{\rm C} = \text{\{$({R}_1, {R}_2)$: \eqref{boundsumcap} and \eqref{eqsumcap2} hold with equality simultaneously\}} \nonumber, \\
  \mathcal{K}_{\rm D} = \text{\{$(R_1,R_2)$: \eqref{cutset2} and \eqref{boundsumcap} hold with equality simultaneously\}} \nonumber, \\
  \mathcal{K}_{\rm E} = \text{\{$(R_1,R_2)$: \eqref{cutset1} and \eqref{cutset2} hold with equality simultaneously\}}. \label{eq:definecorner}\\
\end{eqnarray}



\subsubsection{Achievability for the $(1000)$ Feedback Model}
\label{subsec:(1000)}
The outer-bounds for the $(1000)$ feedback model are given by
\eqref{cutset1}, \eqref{cutset2}, \eqref{boundsumcap} and
\eqref{eqsumcap2}. It is easy to verify that in the weak interference
regime, among the corner points, the corner points $\mathcal{K}_{\rm
  A}$, $\mathcal{K}_{\rm C}$ and $\mathcal{K}_{\rm D}$ form the tightest
outer bound, while in the strong interference regime, among the corner
points, $\mathcal{K}_{\rm B}$ and $\mathcal{K}_{\rm D}$ describe the
tightest outer bound. The achievability is as follows:

 \emph{Encoding:} The messages to be transmitted at both the transmitters are split into common and private
 parts. The common and private messages transmitted from the $u^{\rm
   th}$ transmitter, ${\sf T}_u$, after encoding as channel inputs are
 denoted as $X_{ui,c}$ and $X_{ui,p}$. The corresponding rates are
 $R_{uc}$ and $R_{up}$. The common message generated at ${\sf T}_2$,
 $X_{2i-1,c}$, is learned by ${\sf T}_1$ through feedback before the
 $i^{\rm th}$ block of transmission and re-transmitted by ${\sf T}_1$
 in the $i^{\rm th}$ block. By re-transmitting ${\sf T}_2$'s common
 message, ${\sf T}_1$ performs a relaying action. The achievable rates
 are given by $R_1 =R_{1c} +R_{1p}$ and $R_2 =R_{2c} +R_{2p}$. The
 encoding runs for $B \to \infty$ blocks.


\emph{Decoding:} To ensure reliable decoding, ${\sf T}_1$ remains
silent in the first block and ${\sf T}_2$ remains silent in the last
block. At ${\sf D}_1$, backward decoding is applied, where the common
message of ${\sf T}_2$ is decoded starting $B^{\rm th}$ block. Thus,
at ${\sf D}_1$, before the $i^{\rm th}$ block is decoded $X_{2i,c}$ is
known. Then $X_{2i,c}$ is subtracted from the received message
$Y_{1i}$, assisting in decoding $X_{1i,c}$, $X_{2i-1,c}$ and
$X_{1i,p}$. At ${\sf D}_2$, since forward decoding is applied,
decoding in the $(i-1)^{\rm th}$ block yields $X_{2i-1,c}$ which is
then subtracted from the received message $Y_{2i}$ to decode
$X_{1i,c}$, $X_{2i,c}$ and $X_{2i,p}$. In Fig.~\ref{fig:ach35bf} two
intermediate blocks of received messages at the two receivers are
shown.


\begin{figure}[t]
\centering
\subfigure{\label{fig:d1}
\resizebox{2.3in}{!}{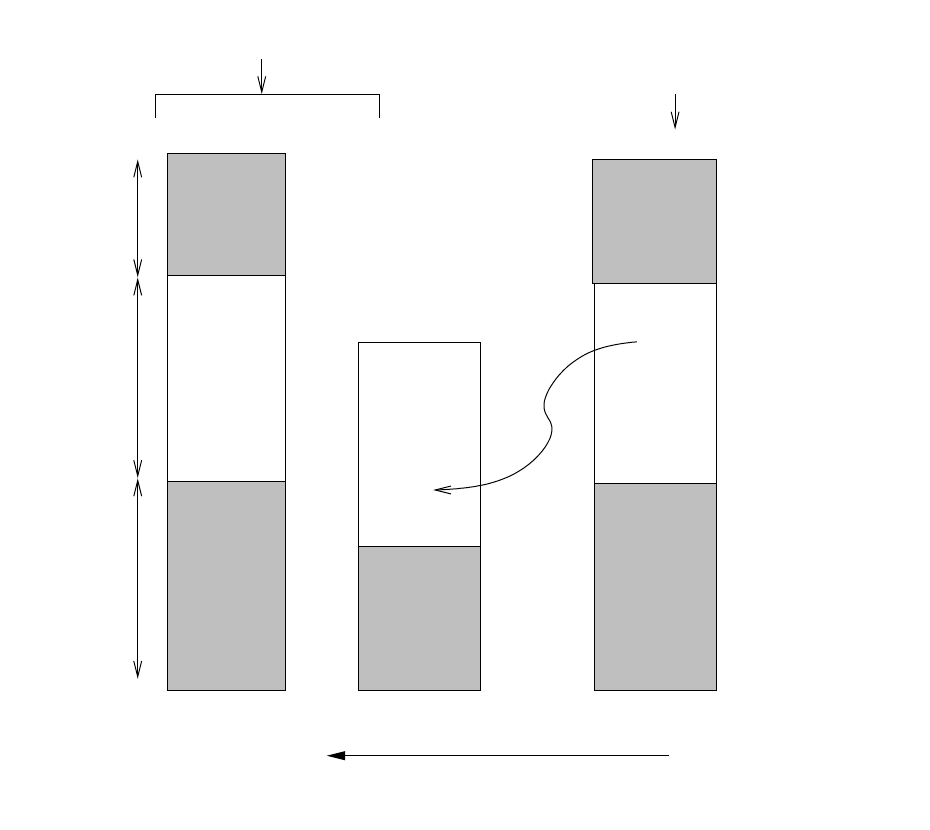}} \hspace{2cm}
\subfigure{\label{fig:d2}
\resizebox{2.0in}{!}{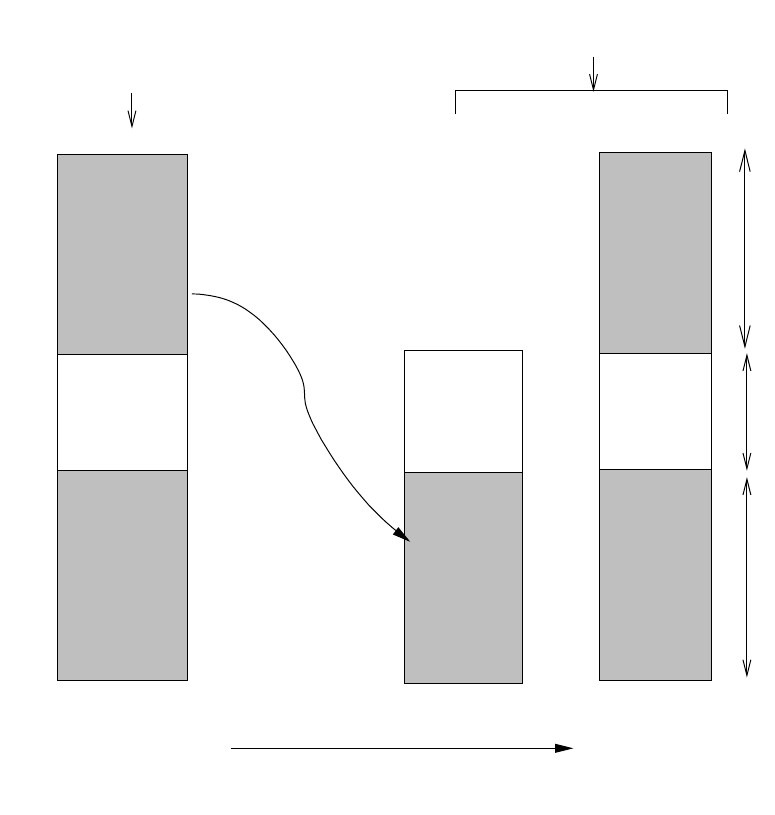}}
   \caption {Achievability of $(R_1,R_2)= {(m, 2n - 2m)}$ with a single direct-link
     feedback. It lies on the boundary of the sum-rate upper bound
     (Lemma \ref{th:sumcap}). At either of the receivers, signals
     learned after being decoded can be subtracted out to
     further decode the rest of the signals.}
   \label{fig:ach35bf}
\end{figure}

\emph{Rate allocation:} In the weak interference regime, the corner
point $\mathcal{K}_{\rm D} \equiv (n-m,n)$ is not achievable without
feedback~\cite{costa82,bresler2,etkin}. Using the achievability
described above, the rate pair $(n-m,n)$ is achievable. In the weak
interference regime, where $\frac{m}{n} \leq 1$, both transmitters
transmit $(n- m)$ bits of private message. Transmitter ${\sf T}_2$
additionally transmits $m$ bits of common message. As $B \to \infty$,
the rates
\begin{equation}
  R_1 = \underbrace{n - m}_{\rm private}, \text{ } R_2 = \underbrace{n - m}_{\rm private} + \underbrace{m}_{\rm common} = n
\end{equation}
are achievable. The corner point $\mathcal{K}_{\rm C} \equiv (m, 2n
-2m)$ is achievable without any feedback except when $\frac{1}{2} \leq
\frac{m}{n} \leq \frac{2}{3}$. When $\frac{1}{2} \leq \frac{m}{n} \leq
\frac{2}{3}$, in order to achieve the rate-pair $(m, 2n - 2m)$, the
private and common message rates are set as
\begin{equation}
R_1 = \underbrace{n - m}_{\rm private} + \underbrace{2m - n}_{\rm
  common} = m,\text{ } R_2 = \underbrace{n-m}_{\rm private} + \underbrace{n -
  m}_{\rm common}  = 2n -2m.
\end{equation}
The corner point $\mathcal{K}_{\rm A}$ is achievable without any
feedback \cite{bresler2,etkin,costa82}.

In the strong interference regime, the corner point $\mathcal{K}_{\rm
  D} \equiv (0,m)$ is not achievable without feedback. However, with
direct-link feedback, the $\mathsf{T}_1$ - $\mathsf{D}_1$ pair along
with the feedback link, can be used as virtual relay node. More
precisely, setting the rates
\begin{equation}
  R_1 = 0, \text{ } R_2 = \underbrace{m}_{\rm common}
\end{equation}
as $B \to \infty$, the rate pair $(0,m)$ bits per block can be
achieved. The other non-trivial corner point $\mathcal{K}_{\rm B}
\equiv (n, m-n)$ is not achievable without feedback when $m >
2n$. Using feedback, when $m > 2n$, by setting rates
\begin{equation}
  R_1 = \underbrace{n}_{\rm common}, \text{ } R_2 =
  \underbrace{m - n}_{\rm common}.
\end{equation}
the desirable rate pair is achievable as $B \to \infty$. The detailed rate
allocation strategy is described in Appendix~\ref{apd_a}, which shows the
achievability of the capacity region of the $(1000)$ feedback model shown
in Table~\ref{table:det_capreg}. \\

\subsubsection{Relating $\mathcal{C}^{(1000)}$ and $\mathcal{C}^{(0010)}$}
In order to re-use the achievability for the $(1000)$ feedback model
described above, in feedback models, which do not have a direct-link
feedback, we show that the capacity regions satisfy $\mathcal{C}^{(1000)} \subseteq \mathcal{C}^{(0010)}$.

\begin{lemma}\label{ThCrslnkD}
The capacity region of the single cross-link feedback and single
direct-link feedback, for $n \geq m$, are related as
\begin{eqnarray}
\mathcal{C}^{(1000)} &\subseteq
&\mathcal{C}^{(0010)} \label{eq:crslnk_d1} \\ \mathcal{C}^{(0001)}
&\subseteq &\mathcal{C}^{(0100)}. \label{eq_crslnk_d2}
\end{eqnarray}
\end{lemma}
\begin{proof}
Due to the symmetry of the channel, it is sufficient to prove only one of
the above inequalities. We prove (\ref{eq:crslnk_d1}). For the single
direct-link feedback model, $(1000)$, the encoding is constrained such
that
\begin{equation} \label{enc_direct}
X_{1i} = f_{1i} (W_1, Y_1^{i-1}), \text{ }X_{2i} = f_{2i} (W_1).
\end{equation}
In the cross link feedback model
\begin{equation} \label{enc_crs}
X_{1i} = g_{1i} (W_1, Y_2^{i-1}), X_{2i} = g_{2i} (W_1).
\end{equation}
Here, $Y_2^{i-1} = X_2^{i-1} \oplus V_1^{i-1}$, where $V_{1i} =
\mathbf{S}^{q-m}X_{1i}$ is the interfering part of the transmitted
message from $\mathsf{T_1}$. Since $X_1^{i-1}$ is known to
$\mathsf{T_1}$ before the $i^\mathrm{th}$ block of encoding,
$V_1^{i-1}$, which is a subset of $X_1^{i-1}$ is also known to
$\mathsf{T}_1$. With the cross-link feedback, since $\mathsf{T}_1$ has
access to $Y_2^{i-1}$ before the $i^\mathrm{th}$ block of encoding, it
can obtain $X_2^{i-1}$. Now, $V_2^{i-1}$ is a subset of $X_2^{i-1}$
(since $m \leq n$), and $Y_1^{i-1} = X_1^{i-1} \oplus
V_2^{i-1}$. Thus, knowing $Y_2^{i-1}$, $\mathsf{T_1}$ can form
$Y_1^{i-1} = (X_{1}^{i-1} \oplus \mathbf{S}^{n-m}(Y_2^{i-1} \oplus
\mathbf{S}^{n-m}X_1^{i-1}))$. Thus, for every message pair $(W_1, W_2)$,
and encoding function $(f_{1i}, f_{2i})$, choosing $g_{1i} \equiv
f_{1i}$ and $g_{2i} \equiv f_{2i}$, the encoding operations defined in (\ref{enc_direct}) and
(\ref{enc_crs}) can be made identical. Identical decoding naturally
follows. Therefore,
\begin{equation}
\mathcal{C}^{(1000)} \subseteq \mathcal{C}^{(0010)}.
\end{equation}
\end{proof}

\begin{remark}
The result in Lemma~\ref{ThCrslnkD} is based on the simple observation
that when $n \geq m$ in the $(0010)$ feedback model, the transmitter
$\mathsf{T}_1$ receives a ``better'' copy of the message encoded at
$\mathsf{T}_2$ than in the $(1000)$ feedback model. This is because,
in the $(0010)$ feedback model, the feedback is received from
$\mathsf{D}_2$, while in the $(1000)$ model, feedback is received from
$\mathsf{D}_1$. At $\mathsf{D}_2$ and $\mathsf{D}_1$, the received
signals are linear combinations of $X_{1i}$ and $X_{2i}$. At
$\mathsf{T}_1$, $X_{1i}$ is known. As $n \geq m$, the bits of $X_{2i}$
that can be decoded from the received message at $\mathsf{D}_2$ form a
superset of the bits of $X_{2i}$ that can be decoded from the received message at
$\mathsf{D}_1$.
\end{remark}

\begin{corollary} \label{corol} When $n \geq m$,
the capacity regions $\mathcal{C}^{(1001)}$, $\mathcal{C}^{(1100)}$,
$\mathcal{C}^{(0110)}$ are related as follows
\begin{eqnarray}\label{eqcorollary}
\mathcal{C}^{(1001)} & \subseteq & \mathcal{C}^{(1100)} \label{eq:equiv1001_1100} \\
\mathcal{C}^{(1001)} & \subseteq &  \mathcal{C}^{(0110)}. \label{eq:equiv1001_0110}
\end{eqnarray}
\end{corollary}
\begin{proof}
In the $(1001)$ feedback model, before the $i^{\rm th}$ block of encoding,
${\sf T}_1$ and ${\sf T}_2$ have access to $Y_1^{i-1}$ and $Y_2^{i-1}$
through feedback. In the $(1100)$ feedback model, ${\sf T}_1$ has access
to $Y_1^{i-1}$ and ${\sf T}_2$ has also access to $Y_1^{i-1}$ before the
$i^{\rm th}$ block of encoding. As shown in Lemma~\ref{ThCrslnkD}, in
the weak interference regime, using $Y_1^{i-1}$, ${\sf T}_2$ can
construct $Y_2^{i-1}$. Therefore, the achievable rate-region of the
$(1100)$ feedback model is at least as large as that of the $(1001)$
feedback model. Thus, $\mathcal{C}^{(1001)} \subseteq
\mathcal{C}^{(1100)}$. Similar proof follows for
\eqref{eq:equiv1001_0110}.
\end{proof}\text{}\\

\subsubsection{Achievability for the $(0010)$ Feedback Model}
The outer bound on the $(0010)$ feedback model is characterized by the
corner points $\mathcal{K}_{\rm B}$ and $\mathcal{K}_{\rm D}$. From
Theorem~\ref{th:detcapreg}, we know that in the weak interference
regime, the corner point $\mathcal{K}_{\rm B} \equiv (n, n -m)$ is
outside the the boundary of $\mathcal{C}^{(1000)}$. Thus the
achievability of the $(1000)$ feedback model and Lemma~\ref{ThCrslnkD}
is not sufficient to show the achievability of the rate pair $(n,n-m)$
for the $(0010)$ feedback model. Therefore, we show a new
achievability based on block-Markov encoding and dirty paper
encoding/decoding to achieve the rate-pair $(n,n-m)$ for the $(0010)$
feedback model.

\emph{Encoding:} The encoding strategy is shown in
Table~\ref{table:encoding_dp}. At ${\sf T}_1$, there is no splitting
of messages. At ${\sf T}_2$, in the $i^{\rm th}$ block, the message is
split into two parts $X_{2i,d}$ and $X_{2i,nd}$. Also, in the $i^{\rm
  th}$ block $X_{2i-1,d}$ is transmitted by ${\sf T}_2$ such that it is
decodable at ${\sf D}_2$ right-away. The message $X_{2i-1,d}$ is known
at ${\sf T}_1$ before the $i^{\rm th}$ block of transmission due to the
cross-link feedback. Therefore, ${\sf T}_1$ can employ a dirty paper
coding like strategy to allow its receiver to decode in the presence of
interference $X_{2i-1,d}$ as shown in
Table~\ref{table:encoding_dp}.

\begin{table}[t]
\centering
  \caption{Encoding of messages in the weak interference regime for the $(0010)$ feedback model}
  \begin{tabular}{ | c | c | c | c |}
    \hline
    & Block 1 & Block $i$ & Block $B$   \\ \hline
    Message $X_{1i}$ at ${\sf T}_1$ & $[X_{11}]$ & $[X_{1i} \oplus \mathbf{S}^{n-m}X_{2i-1}]$ &  $\mathbf{0}_n^T$  \\ \hline
    Message $X_{2i}$ at ${\sf T}_2$ & $[\mathbf{0}_{m}^T, X_{21,nd}^T, X_{21,d}^T]^T$ & $[X_{2i-1,d}^T, \mathbf{0}_p^T, X_{2i,nd}^T, X_{2i,d}^T]^T$  & $[X_{2B-1,d}^T, \mathbf{0}_p^T, X_{2B,nd}^T, X_{2B,d}^T]^T$ \\
    \hline
  \end{tabular}
      \label{table:encoding_dp}
\end{table}

\emph{Decoding:} The messages $X_{1i}$ are decodable at
${\sf D}_1$ and $X_{2i-1,d}$ and $X_{2i,nd}$ are decodable at ${\sf D}_2$
in the $i^{\rm th}$ block as long as the cardinality of the messages are
\begin{eqnarray}
  |X_{2i -1,d}|  = \min(n-m,m),\text{ } |X_{2i,nd}|  =  \max(n - 2m,0),
\end{eqnarray}
and $p = \max(2m - n, 0)$. As $B \to \infty$, the rate-pair $(n, n -
m)$ is achievable, i.e., the corner point $\mathcal{K}_{\rm B}$ is
achievable. In the weak interference regime, from
Lemma~\ref{ThCrslnkD} and the achievability of the $(1000)$ feedback
model, we know that in the weak interference regime, the corner point
$\mathcal{K}_{\rm D} \equiv (n -m,n)$ is achievable with the $(0010)$
feedback model.

In the strong interference regime, from Theorem~\ref{lem:det_region},
and the outer-bounds \eqref{cutset1}, \eqref{cutset2} and
\eqref{boundsumcap}, we know that $\mathcal{C}^{(0010)} \equiv
\mathcal{C}^{(0000)}$. Thus, the capacity region characterization of
the $(0010)$ feedback model is complete.

\subsection{Capacity Regions of the rest of the Feedback Models}
\label{subsec:capreg}
For each feedback model, the capacity region is shown by the
achievability of the subset of corner points \eqref{eq:definecorner}
which form the tightest outer bound, among all the corner points. \\

\subsubsection{$\mathbf{(1001)}, \mathbf{(1101)}$ and
  $\mathbf{(1111)}$ Feedback Models} The capacity region of the
$(1001)$ feedback model is given in Theorem~\ref{th:psumcap}. It can
also be derived using the outer-bounds given by \eqref{boundsumcap},
\eqref{cutset1} and \eqref{cutset2}, and showing the achievability by
treating the $(1001)$ feedback model as a combination of the $(1000)$
and $(0001)$ feedback models. The outer-bound of the capacity region
$\mathcal{C}^{(1001)}$ can be sufficiently characterized by the corner
points $\mathcal{K}_{\rm B}$ and $\mathcal{K}_{\rm D}$. We know that
$\mathcal{K}_{\rm D}$ is achievable with $(1000)$ feedback and thus by
symmetry $\mathcal{K}_{\rm B}$ is achievable with $(0001)$
feedback. Since $\mathcal{C}^{(1000)} \subseteq \mathcal{C}^{(1001)}$ and $\mathcal{C}^{(0001)} \subseteq \mathcal{C}^{(1001)}$, we conclude that $\mathcal{K}_{\rm B}$ and
$\mathcal{K}_{\rm D}$ are achievable with the $(1001)$ feedback model.

The corner points $\mathcal{K}_{\rm B}$ and $\mathcal{K}_{\rm D}$ also
sufficiently characterize the outer-bound of the capacity region of
the $(1111)$ feedback model. As $\mathcal{K}_{\rm B}$ and $\mathcal{K}_{\rm D}$ are both achievable with the $(1001)$
feedback model,
\begin{equation}
\mathcal{C}^{(1001)}\label{eq1001}
\equiv \mathcal{C}^{(1111)}.
\end{equation}

\begin{figure}[t]
\subfigure[{$\frac{m}{n} \in \left[0, \frac{1}{2} \right)$, here $\frac{m}{n}= \frac{1}{3}$}]{\label{fig:gull1} \resizebox{2.2in}{!}{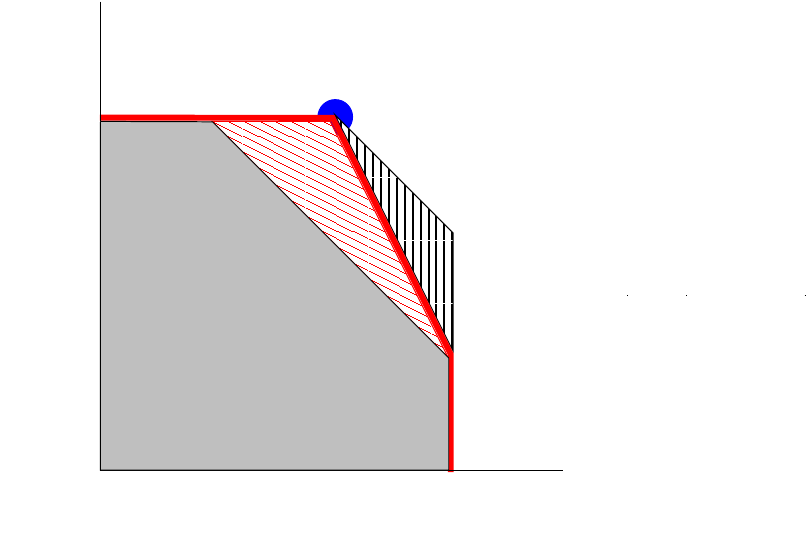}}
\subfigure[{$ \frac{m}{n} \in \left[\frac{1}{2}, \frac{2}{3}\right)$, here $\frac{m}{n} = \frac{3}{5}$}]{\label{fig:gull2}\resizebox{1.8in}{!}{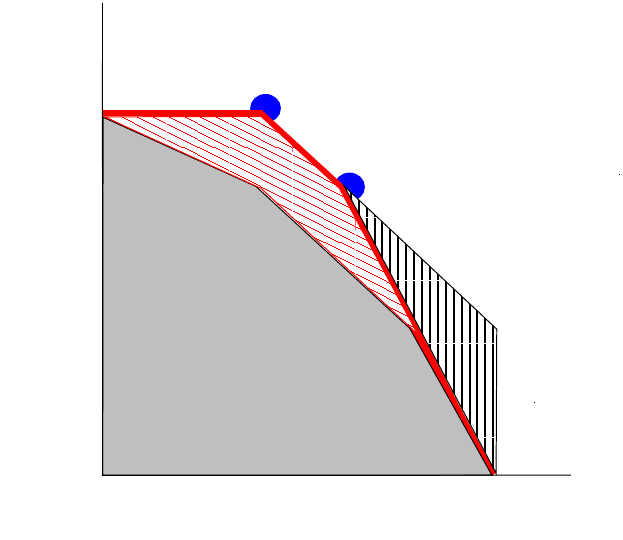}}
\subfigure[{$ \frac{m}{n} \in \left[\frac{2}{3}, 1\right)$, here $\frac{m}{n} = \frac{4}{5}$}]{\label{fig:gull3} \resizebox{1.65in}{!}{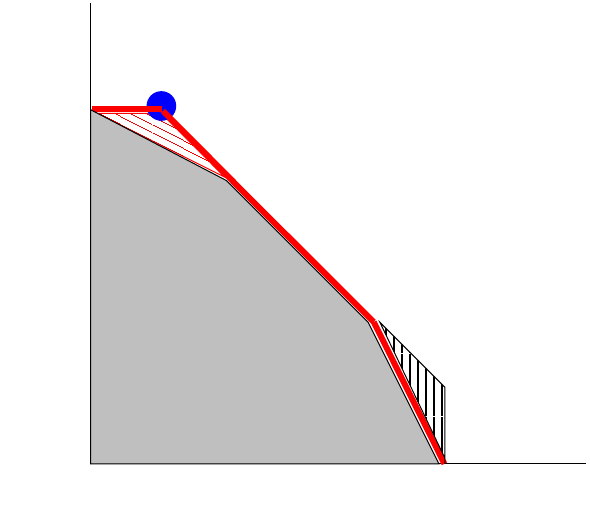}}
\subfigure[{$\frac{m}{n} \in (1,2]$, here~$\frac{m}{n} = \frac{4}{3}$}]{\label{fig:gull4} \resizebox{2.0in}{!}{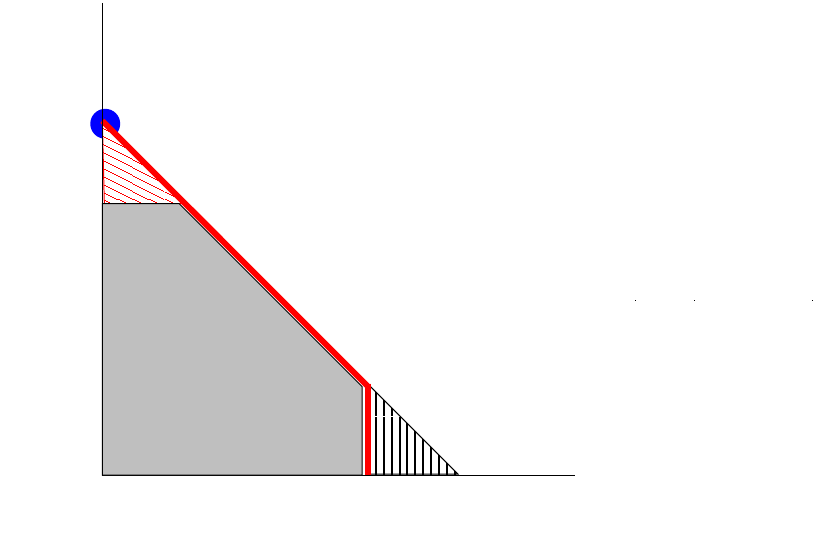}}
\subfigure[{$\frac{m}{n} \in (2, \infty)$, here $\frac{m}{n} = 3$}]{\label{fig:gull5} \resizebox{4.00in}{!}{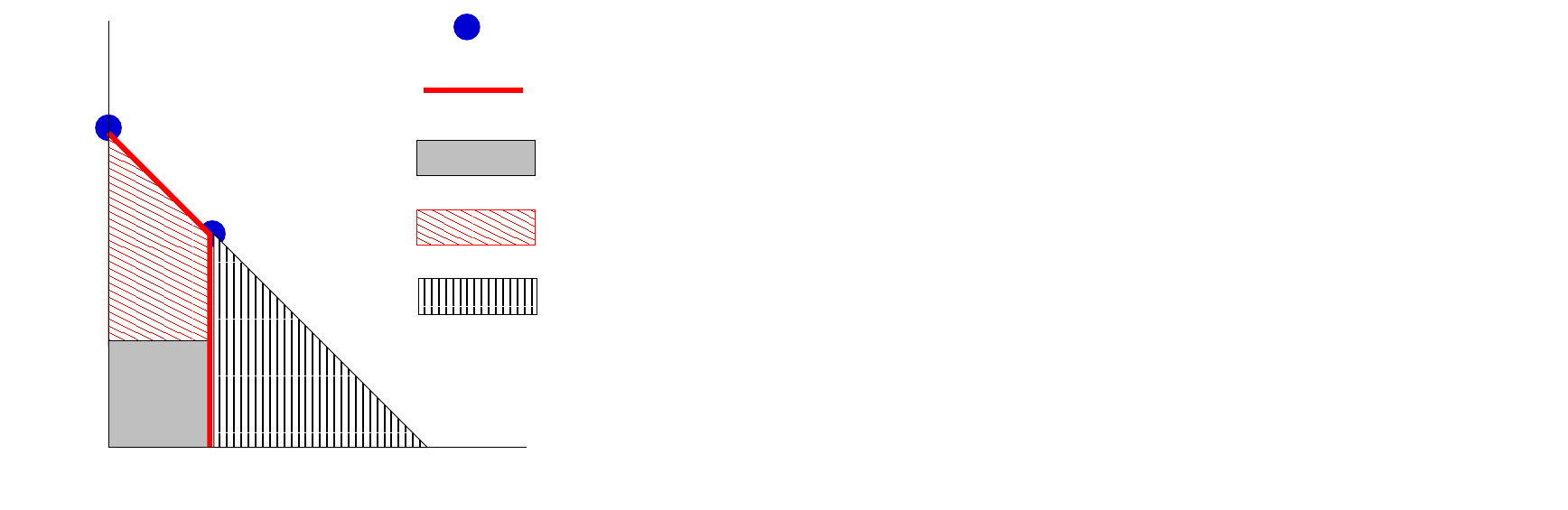}}
   \caption {Typical \emph{normalized} capacity region poly-topes of
     the (1000) feedback model in different regimes of interference. The
     large dots represent the corner points \emph{not}
     achievable without the (1000) feedback model. The figure also
     shows the capacity region of the interference channel with no
     feedback and with the (1111) feedback model.}
   \label{fig:capreg}
\end{figure}

To compare the $(1\mathsf{x}\mathsf{x}1)$ feedback model with the
$(1000)$ feedback model, we note that the point
$\mathcal{K}_{\rm B}$ is not achievable with the latter. This is
because there is no virtual relay path available between
$\mathsf{T}_1$ and $\mathsf{D}_1$. In the $(1\mathsf{x}\mathsf{x}1)$
feedback model, a virtual relay route between $\mathsf{T}_1$ and
$\mathsf{D}_1$ is available and therefore $\mathcal{K}_{\rm B}$ is
achievable. Hence, $\mathcal{C}^{(1\mathsf{x} \mathsf{x} 1)} \supset
\mathcal{C}^{(1000)}$, and the capacity region of the $(1\mathsf{x}
\mathsf{x}1)$ feedback model is strictly larger than the capacity
region of the (1000) feedback model. \\

\subsubsection{$\mathbf{(1100)}, \mathbf{(1110)}$ and $\mathbf{(1010)}$ Feedback Models} We know that
\begin{equation} \label{eq1100}
\mathcal{C}^{(1100)} \subseteq \mathcal{C}^{(1110)} \subseteq
\mathcal{C}^{(1111)}.
\end{equation}
In the weak interference regime, where $n \geq m$, we know from
Corollary~\ref{corol} that $\mathcal{C}^{(1100)} \supseteq
\mathcal{C}^{(1001)}$. Sandwiching the capacity regions
$\mathcal{C}^{(1100)}$ and $\mathcal{C}^{(1110)}$ in between
$\mathcal{C}^{(1001)}$ and $\mathcal{C}^{(1111)}$, from
(\ref{eq1001}), (\ref{eq1100}) and (\ref{eqcorollary}), we conclude
that in the weak interference regime
\begin{equation}
\mathcal{C}^{(1001)} \equiv \mathcal{C}^{(1100)} \equiv \mathcal{C}^{(1110)} \equiv
\mathcal{C}^{(1111)}.
\end{equation}
In the strong interference regime, where $n < m$, the corner points
$\mathcal{K}_{\rm B}$ and $\mathcal{K}_{\rm D}$ sufficiently
characterize $\mathcal{C}^{(1100)}$. We know that when $n < m$, $\mathcal{K}_{\rm B}$ and $\mathcal{K}_{\rm D}$ are also achievable
with $(1000)$ feedback. Since
$\mathcal{C}^{(1110)} \supseteq \mathcal{C}^{(1100)} \supseteq \mathcal{C}^{(1000)}$, we can conclude
that in the strong interference regime
\begin{equation}
\mathcal{C}^{(1110)} \equiv \mathcal{C}^{(1100)} \equiv \mathcal{C}^{(1000)}.
\label{eq:(1010)}
\end{equation}

For the $(1010)$ feedback model, the outer bound of the capacity region in
the weak interference regime is characterized by $\mathcal{K}_{\rm B}$
and $\mathcal{K}_{\rm D}$. The corner point $\mathcal{K}_{\rm D}$ is shown to be achievable with $(1000)$ and corner point
$\mathcal{K}_{\rm B}$ is achievable with $(0010)$ feedback
models. Thus, $(1010)$ can achieve both corner points $\mathcal{K}_{\rm
  B}$ and $\mathcal{K}_{\rm D}$. In the weak interference regime,
since $\mathcal{K}_{\rm B}$ and $\mathcal{K}_{\rm D}$ also
characterize the outer-bound of the $(1111)$ feedback model, in
the weak interference regime we have
\begin{equation}
\mathcal{C}^{(1010)} \equiv \mathcal{C}^{(1111)}.
\end{equation}
As $\mathcal{C}^{(1010)}\subseteq\mathcal{C}^{(1110)}$, from
\eqref{eq:(1010)}, we conclude that in the strong interference
regime
\begin{equation}
\mathcal{C}^{(1010)} \equiv \mathcal{C}^{(1110)} \equiv \mathcal{C}^{(1100)} \equiv \mathcal{C}^{(1000)}.
\end{equation} \\

\subsubsection{$\mathbf{(0110)}$ Feedback Models}
The corner points $\mathcal{K}_{\rm B}$ and $\mathcal{K}_{\rm D}$
characterize the outer-bound for the $(0110)$ feedback model as well
as $(0010)$, and since they are achievable with $(0010)$, they are
also achievable with the $(0110)$ feedback model. Thus,
\begin{equation}
\mathcal{C}^{(0010)} \equiv \mathcal{C}^{(0110)}.
\label{eq:(0110)}
\end{equation}
It is noteworthy that in the strong interference regime, from
Theorem~\ref{lem:det_region} and outer-bounds \eqref{cutset1},
\eqref{cutset2} and \eqref{boundsumcap}, it can easily be confirmed
that
\begin{equation}
\mathcal{C}^{(0000)} \equiv \mathcal{C}^{(0110)}. \label{eq:(0110)_2}
\end{equation}

\subsection{Sum-capacity}
\label{subsec:sumcap}
\subsubsection{Equivalence of the sum-capacity of all $\mathbf{(1\mathsf{xxx})}$ Feedback Models}
In the achievability of the $(1000)$ feedback model, we showed that in the
weak interference regime, the rate pair $(n-m,n)$ and in the strong
interference regime, the rate pair $(0,m)$ is achievable. These rate pairs $(n-m,n)$ and $(0,m)$ both lie on the
outer-bound of the sum-capacity \eqref{boundsumcap}. Thus, using the achievability of the $(1000)$ feedback
model, we can show the achievability of the rate pairs $(n-m,n)$ and
$(0,m)$ for $\mathsf{(1xxx)}$ feedback models, which proves the result
$C_{\rm sum}^{(1\mathsf{xxx})} = C_{\rm sum}^{(1000)}$.\\

\subsubsection{Sum-capacity of the $\mathbf{(0110)}$ and $\mathbf{(0010)}$ Feedback Models}
From Lemma~\ref{th:sumcap}, in the weak interference regime $(n-m,n)$
lies on the outer-bound of the $(1111)$ feedback model. Hence, it is
sum-capacity achieving for $(0010)$ as well as $(0110)$ feedback
models. From Lemma~\ref{ThCrslnkD}, in the weak interference regime
any rate pair that is achievable with the $(1000)$ feedback model should also
be achievable with the $(0010)$ feedback model. Since the rate pair
$(n-m,n)$ is achievable with $(1000)$, it is also achievable with
$(0010)$ and subsequently the $(0110)$ feedback model. Therefore, in the
weak interference regime $C_{\rm sum}^{(0010)} = C_{\rm sum}^{(0110)}
= C_{\rm sum}^{(1111)}$.

In the strong interference regime, we know from \eqref{eq:(0110)} and
\eqref{eq:(0110)_2}, that feedback does not improve the capacity
region for the $(0110)$ feedback model and therefore does not improve the
capacity region of $(0010)$ either. Thus, in the strong interference
regime $C_{\rm sum}^{(0010)} = C_{\rm sum}^{(0110)} = C_{\rm
  sum}^{(0000)}$.

\section{Feedback Models: Gaussian channel}
\label{SecGC}

In this section, the approximate Gaussian capacity regions are
derived for all 9 canonical feedback models. First, we derive two new
outer bounds for the $(1111)$ and $(1000)$ feedback models. Then, we show
an achievability based on Han-Kobayashi type message splitting for the
$(1000)$ model. We prove Lemma~\ref{th:crosslink}, which relates the
achievable rate regions of the $(0010)$ and $(1000)$
feedback models. Additionally, we propose a block-Markov and dirty
paper encoding based achievability scheme for the $(0010)$ feedback model. Finally,
using the achievability of the $(1000)$ and $(0010)$ feedback models, we
show the approximate capacity regions for all canonical feedback
models.

\subsection{Outer Bounds for the Gaussian Channel} \label{sec:outerbound_g}
Now, we present two new outer bounds and the cut-set bound for the
two-user interference channel with various feedback states.

\begin{lemma} [\cite{coverbook,prabhakaran-sc}] \label{lem:cutsetg}
 The two-user symmetric Gaussian interference channel with any one of
 the feedback models is outer bounded by
\begin{eqnarray}
 R_1 & \leq & c_1 \label{cutsetg1}  \\ R_2 & \leq &
 c_2, \label{cutsetg2}
\end{eqnarray}
where
\begin{equation}
     c_1 = \begin{cases}
       \log(1 + {\sf SNR}) & \text{if ${\sf T}_2$ receives no direct-link feedback} \\
       \log(1 + {\sf SNR + INR}) & \text{otherwise},
     \end{cases}\end{equation}\begin{equation}
     c_2 = \begin{cases}
       \log(1 + {\sf SNR}) & \text{if ${\sf T}_1$ receives no direct-link feedback} \\
       \log(1 + {\sf SNR + INR}) & \text{otherwise}.
     \end{cases}
\end{equation}
\end{lemma}

The following theorem provides an outer-bound on the sum-capacity of the
$(1111)$ feedback model.
\begin{theorem} \label{th:sum-capg}
The sum capacity of the two-user symmetric Gaussian interference channel for the
$(1111)$ feedback model is outer bounded by
\begin{equation}
\label{eq:sum-capg} R_1 + R_2 \leq \sup_{0 \leq |\rho| \leq 1}  \log \left( 1 +
\frac{(1 - |\rho|^2) \mathsf{SNR}}{1 + (1 - |\rho|^2)
  \mathsf{INR}}\right) + \log (1 + \mathsf{SNR + INR} +
2|\rho|\sqrt{\mathsf{SNR. INR}} ).
\end{equation}
\end{theorem}
\begin{proof}
 The proof details are provided in Appendix \ref{pthscg}.
\end{proof}
\begin{remark}
As the sum-capacity of the $(1111)$ feedback model is an outer bound
on the sum-capacity of all feedback models, Theorem~\ref{th:sum-capg} also applies as an outer bound on the
sum-capacity of all feedback models.
\end{remark}
Note that \eqref{eq:sum-capg} can further be upper bounded to yield
\begin{equation}\label{eq:sum-capg_norho}
  R_1 + R_2 \leq \log \left( 1 +
  \frac{\mathsf{SNR}}{1 + \mathsf{INR}}\right) + \log (1 + \mathsf{SNR + INR} +
2\sqrt{\mathsf{SNR. INR}} ).
\end{equation}

As observed in the deterministic case, the bound on the sum-capacity is not sufficient to describe the capacity region of the $(1000)$
feedback model. The following theorem is an upper bound on the rate
$2R_1 + R_2$.
\begin{theorem}\label{th:sum-capg2}
The capacity region of the two-user symmetric Gaussian interference channel with feedback state $(1000)$ is outer bounded by
\begin{eqnarray}\label{eq:sum-capg2} 2R_1 + R_2 & \leq &
  \sup_{0 \leq |\rho| \leq 1} \log \left( 1 + \frac{(1 - |\rho|^2)
    \mathsf{SNR}}{1 + (1 - |\rho|^2) \mathsf{INR}}\right) + \log (1 +
  \mathsf{SNR + INR} + 2|\rho|\sqrt{\mathsf{SNR. INR}} )  \nonumber
  \\ && + \log\left(1 + \mathsf{INR} + \frac{\mathsf{SNR} - (1 +
    |\rho|^2) \mathsf{INR} + 2 |\rho| \sqrt{\mathsf{SNR.INR} }}{1+
    \mathsf{INR}}\right) .
\end{eqnarray}
\end{theorem}
\begin{proof}
The proof details are provided in Appendix \ref{proof_th_sum-capg2}.
\end{proof}

To characterize the approximate capacity region of the $(1000)$ feedback model, we
will use the bound in \eqref{eq:sum-capg2} only in the weak
interference regime. In the weak interference regime, an upper bound
for \eqref{eq:sum-capg2} is
\begin{eqnarray} 2R_1 + R_2  \leq
  \log \left( 1 + \frac{
    \mathsf{SNR}}{1 + \mathsf{INR}}\right) + \log (1 +
  \mathsf{SNR + INR} + 2\sqrt{\mathsf{SNR. INR}} ) +
   \log\left(1 + \mathsf{INR} + \frac{\mathsf{SNR} - \mathsf{INR}}{1+
    \mathsf{INR}}\right)  \label{eq:sum-capg2_norho1}
\end{eqnarray}

In Table~\ref{table:gauss_region}, we present the approximate capacity regions of the different feedback models
studied in this paper. The table also lists the gap to capacity for each of the feedback
models. These gaps are computed for the achievability schemes that will be described in
Section \ref{sec:gauss-ach}.

\begin{table}[h]
\caption{Approximate capacity regions for all 9 canonical feedback models}
\centering
\begin{tabular}{|c|l|c|}
  \hline && \\
  Cases & Outer bound of Capacity Region  & Gap to Capacity \\ && \\
  \hline & & \\
     & $R_1\le \log(1 + \mathsf{SNR} + \mathsf{INR})$  &\\
    & $R_2\le \log(1 + \mathsf{SNR} + \mathsf{INR})$  & $2.59$ bits/Hz \\
  $(1\mathsf{x}\mathsf{x}1)$ &$R_1+R_2\le \sup_{0 \leq |\rho| \leq 1} \{ \log \left( 1 +
 \frac{(1 - |\rho|^2) \mathsf{SNR}}{1 + (1 - |\rho|^2)\mathsf{INR}}\right) +$  &\\
 & $\log (1 + \mathsf{SNR + INR} +
 2|\rho|\sqrt{\mathsf{SNR. INR}} )\}$ & \\ && \\ \hline
    &  & \\
     &  $R_1\le \log(1 + \mathsf{SNR})$ & $2.59$~bits/Hz  \\
   $(1100)$, $(1110)$     & $R_2\le \log(1 + \mathsf{SNR} + \mathsf{INR})$  & for $(1100)$ and $(1110)$ \\
 $(1010)$  &$R_1+R_2\le \sup_{0 \leq |\rho| \leq 1} \{ \log \left( 1 +
 \frac{(1 - |\rho|^2) \mathsf{SNR}}{1 + (1 - |\rho|^2)\mathsf{INR}}\right) +$  & \\
  & $\log (1 + \mathsf{SNR + INR} +
 2|\rho|\sqrt{\mathsf{SNR. INR}} )\}$ & $4.59$~bits/Hz for $(1010)$
\\ && \\
    \hline &  &\\
& $R_1\le \log(1 + \mathsf{SNR})$  & $2.59$~bits/Hz for $(0110)$  \\
& $R_2\le \log(1 + \mathsf{SNR})$  &
  \\
$(0110),(0010)$ & $R_1+R_2\le \sup_{0 \leq |\rho| \leq 1} \log \left( 1 +
 \frac{(1 - |\rho|^2) \mathsf{SNR}}{1 + (1 -
   |\rho|^2)\mathsf{INR}}\right) + $  & $4.59$~bits/Hz for $(0010)$ \\
& $\log (1 + \mathsf{SNR + INR} +
 2|\rho|\sqrt{\mathsf{SNR. INR}} )$ &    \\  & &\\
\hline & & \\
     &  $R_1\le \log(1 + \mathsf{SNR})$ & \\
      & $R_2\le \log(1 + \mathsf{SNR} + \mathsf{INR})$ &\\
    $(1000)$  &$R_1+R_2\le \sup_{0 \leq |\rho| \leq 1} \{ \log \left( 1 +
 \frac{(1 - |\rho|^2) \mathsf{SNR}}{1 + (1 - |\rho|^2)\mathsf{INR}}\right) +$ &  \\
 & $\log (1 + \mathsf{SNR + INR} +
 2|\rho|\sqrt{\mathsf{SNR. INR}} )\}$  & $4.59$~bits/Hz \\
& $2R_1+R_2\le \sup_{0 \leq |\rho| \leq 1} \{ \log \left( 1 +
 \frac{(1 - |\rho|^2) \mathsf{SNR}}{1 + (1 - |\rho|^2)\mathsf{INR}}\right) +$ & \\
 & $\log (1 + \mathsf{SNR + INR} +
 2|\rho|\sqrt{\mathsf{SNR. INR}} ) + $ & \\
& $\log\left(1 + \mathsf{INR} + \frac{\mathsf{SNR} - (1 +
|\rho|^2) \mathsf{INR} + 2 |\rho| \sqrt{\mathsf{SNR.INR} }}{1+
\mathsf{INR}}\right) \}$  & \\
& & \\
    \hline
  \end{tabular}\label{table:gauss_region}
\end{table}

 \subsection{Achievability}\label{sec:gauss-ach}
In this section, we show the achievability of the sum-rate and the
rate regions, which are within a constant number of bits of the outer
bound developed in Section \ref{sec:outerbound_g}. The achievable rate
region as well as the outer-bound are implicitly parameterized by the
pair $(\mathsf{SNR,INR})$. 
Let the set of all corner points (vertices) of the convex polygon,
which forms the outer bound for feedback state $(F_{11} F_{12} F_{21}
F_{22})$, be denoted by $\mathcal{Q}^{(F_{11} F_{12} F_{21}
  F_{22})}$. Then in order to prove that the capacity region is within
a constant number of bits of the outer bound, it is sufficient to
prove
\begin{equation}
\max_{\mathsf{SNR, INR}}
\min_{\mathcal{R}^{(F_{11} F_{12} F_{21}
F_{22})}} \max(\overline{C}_1 -
R_1, \overline{C}_2 - R_2) \leq c,
\end{equation}
where $(R_1,R_2) \in \mathcal{R}^{(F_{11}F_{12}F_{21}F_{22})}$, $(\overline{C}_1, \overline{C}_2) \in \mathcal{Q}^{(F_{11}
  F_{12} F_{21} F_{22})}$ and $c$ is a fixed constant independent of $\mathsf{SNR}$ and $\mathsf{INR}$. Therefore, in this section, for each corner
point on the outer bound, we show an achievable rate pair that is
within $c$ bits from it. The corner points of relevance are defined
here as
\begin{eqnarray}
  \overline{\mathcal{K}}_{\rm A}= \{(\overline{C}_1, \overline{C}_2):
  \text{$\overline{C}_1 = R_1$ \&  $\overline{C}_2 = R_2$ such that \eqref{cutsetg1} and \eqref{eq:sum-capg2_norho1} hold with equality simultaneously}\},
  \nonumber \\
  \overline{\mathcal{K}}_{\rm B}= \{(\overline{C}_1, \overline{C}_2):
  \text{$\overline{C}_1 = R_1$ \&  $\overline{C}_2 = R_2$ such that \eqref{cutsetg1} and \eqref{eq:sum-capg_norho} hold with equality simultaneously}\},\nonumber   \\
  \overline{\mathcal{K}}_{\rm C}= \{(\overline{C}_1, \overline{C}_2):
  \text{$\overline{C}_1 = R_1$ \&  $\overline{C}_2 = R_2$ such that \eqref{eq:sum-capg_norho} and \eqref{eq:sum-capg2_norho1} hold with equality simultaneously}\}, \nonumber  \\   \overline{\mathcal{K}}_{\rm D}= \{(\overline{C}_1, \overline{C}_2):
  \text{$\overline{C}_1 = R_1$ \&  $\overline{C}_2 = R_2$ such that \eqref{cutsetg2} and \eqref{eq:sum-capg_norho} hold with equality simultaneously}\},\nonumber
  \\
  \overline{\mathcal{K}}_{\rm E}= \{(\overline{C}_1, \overline{C}_2):
  \text{$\overline{C}_1 = R_1$ \&  $\overline{C}_2 = R_2$ such that \eqref{cutsetg1} and \eqref{cutsetg2} hold with equality simultaneously}\}.  \label{eq:definecornerg}
\end{eqnarray}
Note that $\overline{\mathcal{K}}_{\rm A}$ and
$\overline{\mathcal{K}}_{\rm C}$ are defined only for the weak
interference regime as we will need to show achievable rate pairs
within constant number of bits from them only in the weak interference
regime. Moreover, note that for a fixed ${\sf SNR, INR}$ the rate pair
described by a corner point in the outer bound will change based on
the feedback model, since the bounds \eqref{cutsetg1} and
\eqref{cutsetg2} vary based on the feedback model.

We next describe the achievability for the $(1000)$ and $(0010)$ feedback
models, find $\mathcal{R}^{(1000)}$ and $\mathcal{R}^{(0010)}$, and
then use them to obtain the approximate capacity regions of all 9
canonical feedback models.\\

\subsubsection{Achievability for the $\mathbf{(1000)}$ Feedback Model}
\label{subsec:onelink_a_l}
\paragraph{{\bf Weak Interference}, $\alpha \leq 1$} The outer-bound of the capacity region of the
$(1000)$ feedback model is sufficiently characterized by
$\overline{\mathcal{K}}_{\rm A}$, $\overline{\mathcal{K}}_{\rm C}$ and
$\overline{\mathcal{K}}_{\rm D}$. To achieve within constant number of
bits of $\overline{\mathcal{K}}_{\rm A}$, feedback is not required,
while to achieve within a constant number bits of
$\overline{\mathcal{K}}_{\rm C}$ and $\overline{\mathcal{K}}_{\rm D}$,
feedback is needed.

\emph{Encoding}: Similar to the
achievability in Section~\ref{SecDC}, we use the Han-Kobayashi
rate-splitting approach \cite{han}. At both transmitters the message
to be transmitted is split into common and private parts. The common
message generated by ${\sf T}_2$ in the $i^{\rm th}$ block is learned
by ${\sf T}_1$ via feedback, decoded, re-encoded and re-transmitted
in the $(i+1)^{\rm th}$ block. In the $i^\mathrm{th}$ block of
transmission, the common and private messages generated by the
$u^\mathrm{th}$ transmitter are denoted by $X_{ui,c}$ and $X_{ui,p}$
respectively. Rates $R_{up}$, $R_{uc}$ denote the private and common rates for the
${\sf T}_u - {\sf D}_u$ pair. Thus, $R_u = R_{up} + R_{uc}$.
The fraction of power allocated to the common and
private parts is $\lambda_{uc}$ and $\lambda_{up}$. To transmit the
common message of ${\sf T}_2$, ${\sf T}_1$ allocates $\lambda_{1r}$
fraction of its power. The power constraint implies $\lambda_{1c} +
\lambda_{1p} + \lambda_{1r} \leq 1$ and $\lambda_{2c} + \lambda_{2p}
\leq 1$. As a simplification step, we propose $\lambda_{1p} =
\lambda_{2p}$. The following communication strategy, which extends to
$B$ blocks is proposed
\begin{equation}
X_{1i} = \left\{
\begin{array}{cl}
0 &  i = 1 \\ \sqrt{\lambda_{1p}}X_{1i,p} +
\sqrt{\lambda_{1c}}X_{1i,c} + \sqrt{\lambda_{1r}}X_{2i-1,c} & 1< i \leq B
\end{array} \right.
\label{eq:acht1}
\end{equation}
and
\begin{equation}
X_{2i} = \left\{
\begin{array}{cl}
\sqrt{\lambda_{2p}}X_{2i,p} + \sqrt{\lambda_{2c}}X_{2i,c} &  1 \leq i <B \\ 0 &  i = B \\
\end{array} \right.
\label{eq:acht2}
\end{equation}

\emph{Decoding}: We will employ \emph{forward decoding} at
$\mathsf{D_2}$ and \emph{backward decoding} (starting from the
$B^\mathrm{th}$ block) at $\mathsf{D_1}$. Since forward decoding is
used at ${\sf D}_2$, the message $X_{2i - 1,c}$ is decoded before decoding the $i^{\rm th}$ block. Thus,
$g_c\sqrt{\lambda_{1r}}X_{2i-1,c}$ can be subtracted from the received
message $Y_{2i}$ while decoding the messages received in the
$i^\mathrm{th}$ block. On the other hand, at $\mathsf{D_1}$, since
backward decoding is employed, message $X_{2i,c}$ is decoded while
decoding the $(i + 1)^\mathrm{th}$ block of received messages. Thus,
it can be used to subtract out $g_c\sqrt{\lambda_{2c}}X_{2i,c}$ from
the received message $Y_{1i}$ to assist decoding the $i^\mathrm{th}$
block. At ${\sf D}_1$, the private messages $X_{1i,p}$ and $X_{2i,p}$
are treated as noise while decoding $X_{1i,c}$ and $X_{2i-1,c}$. After
decoding $X_{1i,c}$ and $X_{2i-1,c}$, they are subtracted out from
$Y_{1i}$ and $X_{1i,p}$ is decoded treating $X_{2i,p}$ as
noise. Similar steps follow at the receiver ${\sf D}_2$.
At $\mathsf{D}_1$,
the decoding constraint can be written as
\begin{eqnarray}
 \label{eq:11} R_{1c} & \leq & \log\left(1  + \frac{\lambda_{1c}{\sf SNR}}{\lambda_{1p}{\sf SNR} + \lambda_{2p}{\sf INR} + 1}\right) \\
 \label{eq:12} R_{2c} & \leq & \log\left(1  + \frac{\lambda_{1r}{\sf SNR}}{\lambda_{1p}{\sf SNR} + \lambda_{2p}{\sf INR} + 1}\right) \\
\label{eq:13}  R_{1c} + R_{2c} & \leq & \log\left(1  + \frac{(\lambda_{1c} + \lambda_{1r}){\sf SNR}}{\lambda_{1p}{\sf SNR} + \lambda_{2p}{\sf INR} + 1}\right) ,
\end{eqnarray}
while at ${\sf D}_2$, the decoding constraints are
\begin{eqnarray}
\label{eq:21}  R_{1c} & \leq & \log\left(1  + \frac{\lambda_{1c}{\sf INR}}{\lambda_{1p}{\sf SNR} + \lambda_{2p}{\sf INR} + 1}\right) \\
\label{eq:22}  R_{2c} & \leq & \log\left(1  + \frac{\lambda_{2c}{\sf SNR}}{\lambda_{1p}{\sf SNR} + \lambda_{2p}{\sf INR} + 1}\right) \\
\label{eq:23}  R_{1c} + R_{2c} & \leq & \log\left(1  + \frac{\lambda_{1c}{\sf INR} + \lambda_{2c}{\sf SNR}}{\lambda_{1p}{\sf SNR} + \lambda_{2p}{\sf INR} + 1}\right).
\end{eqnarray}

Further, since we are employing a decode and forward kind of strategy
for re-transmitting $X_{2i-1,c}$, before forwarding it, $\mathsf{T_1}$
has to decode it using the signal $(Y_{1i} - g_dX_{1i})$ ($X_{1i}$
available via feedback). This imposes the following decoding
constraints
\begin{eqnarray}\label{eq:2mac}
R_{2c} & \leq &\log\left( 1 + \frac{\lambda_{2c} \mathsf{INR}}{\lambda_{2p}
\mathsf{INR} + 1}\right),
\end{eqnarray}

Finally, the decoding constraints for the private messages are
\begin{eqnarray}
R_{1p} & \leq & \log \left( 1 + \frac{\lambda_{1p}
\mathsf{SNR}}{\lambda_{2p}\mathsf{INR} + 1} \right)  \label{eq:r1p}\\
R_{2p} & \leq & \log \left( 1 + \frac{\lambda_{2p}
\mathsf{SNR}}{\lambda_{1p}\mathsf{INR} + 1} \right).\label{eq:r2p}
\end{eqnarray}

\emph{Choice of power and rate allocation}: In Tables
\ref{table:lambdas} and \ref{table:rates}, the power and corresponding
rate allocation for four different rate pairs $(R_1,R_2)$, labeled
$\mathcal{P}_{\rm A}, \mathcal{P}_{\rm B}, \mathcal{P}_{\rm C}$ and
$\mathcal{P}_{\rm D}$ are shown. Note that, for each of the rate pairs
labeled by $\mathcal{P}_{\rm A}, \mathcal{P}_{\rm B}, \mathcal{P}_{\rm
  C}$ and $\mathcal{P}_{\rm D}$, we can obtain $R_1 = R_{1p} + R_{1c}$
and $R_2 = R_{2p} + R_{2c}$ from Table \ref{table:rates}. Using the
achievable strategy for the $(1000)$ feedback model, in the weak
interference regime, the rate pairs labeled by $\mathcal{P}_{\rm A},
\mathcal{P}_{\rm C}$ and $\mathcal{P}_{\rm D}$, described in Table
\ref{table:rates}, are easily shown to be feasible for the power
allocation described in Table~\ref{table:lambdas}.

The rate pairs described by $\mathcal{P}_{\rm A}$, $\mathcal{P}_{\rm
  C}$ and $\mathcal{P}_{\rm D}$ in Table \ref{table:rates} are
within a constant number of bits from $\overline{\mathcal{K}}_{\rm A}$,
$\overline{\mathcal{K}}_{\rm C}$ and $\overline{\mathcal{K}}_{\rm D}$
respectively. The gaps of $\mathcal{P}_{\rm A}$, $\mathcal{P}_{\rm C}$
and $\mathcal{P}_{\rm D}$ from $\overline{\mathcal{K}}_{\rm A}$,
$\overline{\mathcal{K}}_{\rm C}$ and $\overline{\mathcal{K}}_{\rm D}$
for the $(1000)$ feedback model are evaluated in
Appendix~\ref{para:2r1r2-r1}, \ref{para:2r1r2-r1r2} and
\ref{para:r2-r1r2_l} and the maximum gap is found to be
$4.59$~bits/Hz.  \\

\begin{table}[t] \caption{Power allocation for the private and common messages for $(1000)$ feedback model}
\label{table:lambdas}
\centering
\begin{tabular}{ |c|l|c|c |c|c|c|}
\hline %
Corner Point  & $\alpha$& $\lambda_{1p}$ & $\lambda_{2p}$ &$\lambda_{1c}$ & $\lambda_{2c}$ & $\lambda_{1r}$ \\\hline %
& $[0,1/2)$ & 1 & $\min(1, 1/{\sf INR})$ & 0  & 0 & 0  \\ \cline{2-7}
$\mathcal{P}_{\rm A}$  & $[1/2, 1]$ & 1 & 0 & 0 & 0 & 0 \\ \hline
 & $(1,2]$ & 0 & 0  & 1 &  1 & 0 \\ \cline{2-7}
$\mathcal{P}_{\rm B}$ & $(2,\infty$) & 0 & 0 & $1 - \frac{1}{\sf SNR}$& $1$ & $\frac{1}{\sf SNR}$  \\ \hline
 & $[0, 1/2)$ & $\min(1,1/{\sf INR})$ & $\min(1,1/{\sf INR})$ & 0 & $1 - \lambda_{1p}$ & $1 - \lambda_{2p}$  \\\cline {2 -7}
$\mathcal{P}_{\rm C}$ & $[1/2,2/3)$ & $\min(1,1/{\sf INR})$ & $\min(1,1/{\sf INR})$ & $\frac{(1 - \lambda_{1p})}{2}$ & $(1 - \lambda_{2p})$ & $\frac{(1 - \lambda_{1p})}{2}$  \\ \cline{2-7}
& $[2/3,1]$ & $\min(1,1/{\sf INR})$ & $\min(1,1/{\sf INR})$ & 1 - $\lambda_{1p}$ & 1 - $\lambda_{2p}$ & 0 \\ \hline
 & $[0,1]$ & $\min(1,1/{\sf INR})$& $\min(1,1/{\sf INR})$ & 0 & $1 - \lambda_{1p}$ & $1 - \lambda_{2p}$   \\ \cline {2-7}
$\mathcal{P}_{\rm D}$ & $(1,\infty)$ & 0 & 0 & 0 &   1  &  1 \\\hline
\end{tabular}
\end{table}

\begin{table}[t]\caption{Rate allocation to the private and common messages for $(1000)$ feedback model}
\label{table:rates}
\centering
\begin{tabular}{ |c|l|c|c |c|c|c|}
\hline %
Corner Point & $\alpha$& $R_{1p}$ & $R_{2p}$ &$R_{1c}$ & $R_{2c}$ \\\hline %
 & $[0,1/2)$ & $\log({\sf SNR}/{2})$ & $\log({\sf SNR}/{\sf 2 INR^2})$ & 0  & 0   \\ \cline{2-6}
$\mathcal{P}_{\rm A}$  & $[1/2, 1]$ & $\log(1 + {\sf SNR})$ & 0 & 0 & 0  \\ \hline
 & $(1,2]$ & 0 & 0  &   $\log({\sf SNR})$ & $\log(1 + \frac{\sf INR}{\sf SNR})$  \\ \cline{2-6}
$\mathcal{P}_{\rm B}$ & (2,$\infty$) & 0 & 0 & $\log({\sf SNR})$& $\log({\sf \frac{INR}{SNR}}) $  \\ \hline
  & $[0, 1/2)$ & $\log(1 + {\sf SNR}/{\sf 2INR} )$& $R_{1p}$ & 0 & $\log({\sf INR}/{3})$  \\\cline {2 -6}
& & & & &  \\
$\mathcal{P}_{\rm C}$ & $[1/2,2/3)$ & $\log(1 + \frac{\sf SNR}{\sf 2INR})$ & $R_{1p}$ & $\log(1 + \frac{\sf INR^2}{\sf SNR}) -2$  & $\log(\frac{1 + {\sf SNR}/{\sf INR}}{4})$   \\
 & & & & &  \\ \cline{2-6}
 & $[2/3,1]$ & $\log(1 + {\sf SNR}/{\sf 2INR})$ & $R_{1p}$  & $\log({\sf INR^2}/{\sf 3SNR})$ & $\log({\sf 2SNR}/{\sf 3INR})$  \\ \hline
 & $[0,1]$ & $\log(1 + \frac{\sf SNR}{2 {\sf INR} })$& $R_{1p}$ & 0 & $\log(\frac{{\sf INR}}{3})$  \\ \cline {2 -6}
$\mathcal{P}_{\rm D}$ & $(1,\infty)$ & 0 & 0  & 0 & $\log(1 + {\sf INR})$   \\\hline
\end{tabular}
\end{table}

\paragraph{{\bf Strong Interference}, $\alpha > 1$} The outer-bound is sufficiently described by
$\overline{\mathcal{K}}_{\rm B}$ and $\overline{\mathcal{K}}_{\rm
  D}$. For $1 < \alpha < 2$, the achievable rate pair described by
$\mathcal{P}_{\rm B}$ in Table~\ref{table:rates} can be achieved
without feedback and is within constant number of bits from
$\overline{\mathcal{K}}_{\rm B}$. The constant is evaluated to be
$2.59$~bits/Hz in Appendix~\ref{para:r1-r_1r_2}. For the rest, the
following achievable strategy is employed.

The encoding is identical with \eqref{eq:acht1} and \eqref{eq:acht2}. In
contrast to the decoding scheme for $\alpha \le 1$, in strong
interference, we employ forward decoding at ${\sf D}_1$ and backward
decoding at ${\sf D}_2$. Private messages are not needed in this
regime, thus $\lambda_{1p} = \lambda_{2p} = 0$, and correspondingly
$R_{1p} = R_{2p} = 0$. Since forward decoding is employed at
$\mathsf{D}_1$, $X_{2i-1,c}$ is decoded from the received message in
the $(i-1)^\mathrm{th}$ block of decoding, and thus
$g_d\sqrt{\lambda_{1r}}X_{2i-1,c}$ can be subtracted out from the
received message $Y_{1i}$ for decoding the $i^{\rm th}$ block. On the
other hand, at $\mathsf{D_2}$, backward decoding is applied. Thus,
prior to decoding the $i^\mathrm{th}$ block, $X_{2i,c}$ is known and
can be used to subtract $g_d\sqrt{\lambda_{2c}}X_{2i,c}$ from $Y_{2i}$.
Then, $X_{2i-1,c}$ and $X_{1i,c}$ are decoded. With
$\lambda_{up} = 0$, the decoding constraints at ${\sf D}_1$, ${\sf
  D}_2$ and ${\sf T}_1$ are the same as \eqref{eq:11}-\eqref{eq:13},
\eqref{eq:21}-\eqref{eq:23} and \eqref{eq:2mac}.

\emph{Choice of power and rate allocation:} In the strong interference
regime, rate pairs described by $\mathcal{P}_{\rm B}$ and
$\mathcal{P}_{\rm D}$ in Table~\ref{table:rates} are feasible for the
power allocation described by Table~\ref{table:lambdas}. The rate pairs
described by $\mathcal{P}_{\rm B}$ and $\mathcal{P}_{\rm D}$ in
Table~\ref{table:rates} are within constant number of bits of
$\overline{\mathcal{K}}_{\rm B}$ and $\overline{\mathcal{K}}_{\rm D}$
respectively, for which the gap is computed in
Appendix~\ref{para:r1-r_1r_2} and \ref{para:r2-r1r2_h}. The maximum
gap is found to be $2.59$~bits/Hz. This completes the characterization
of the approximate capacity for the $(1000)$ feedback within
$4.59$~bits/Hz.\\

\subsubsection{Relating $\mathcal{R}^{(1000)}$ and $\mathcal{R}^{(0010)}$}
For proving the achievability of the rest of the feedback models, we
prove the following lemma, which relates $\mathcal{R}^{(1000)}$ and
$\mathcal{R}^{(0010)}$.
\begin{lemma}\label{th:crosslink}
When $\alpha \leq 1$, there exist achievable rate regions
$\mathcal{R}^{(0010)}$ and $\mathcal{R}^{(0100)}$ such that
\begin{eqnarray}
\label{eq:crslnk1} \mathcal{R}^{(1000)} & \subseteq &\mathcal{R}^{(0010)}  \\
\label{eq:crslnk2} \mathcal{R}^{(0001)} & \subseteq &\mathcal{R}^{(0100)},
\end{eqnarray}
where the region $\mathcal{R}^{(1000)}$ is described for the feedback
model $(1000)$ in Section~\ref{subsec:onelink_a_l}. The region
$\mathcal{R}^{(0001)}$ is such that if $(R_x,R_y) \in
\mathcal{R}^{(1000)}$, then $(R_y,R_x) \in \mathcal{R}^{(0001)}$.
\end{lemma}
\begin{proof}
Due to the symmetry, proving (\ref{eq:crslnk1}) is sufficient. Suppose
that for the $(0010)$ feedback model, the encoding is identical to the
one in the $(1000)$ feedback model. Then, the feasibility of decoding
needs to be established for the $(0010)$ feedback model, given that
decoding is feasible for the $(1000)$ feedback model. Let the decoding at
both the receivers also be identical. Then for a given choice of
$\{R_{1c},R_{2c},R_{1p},R_{2p}\}$ and $\{\lambda_{1c},
\lambda_{2c},\lambda_{1p},\lambda_{2p}, \lambda_{2r}\}$, the decoding
constraints at the receivers for the $(0010)$ feedback model are
identical to (\ref{eq:11})-(\ref{eq:23}) and
\eqref{eq:r1p},\eqref{eq:r2p}, which are known to be feasible for the
$(1000)$ feedback model. The decoding constraints at $\mathsf{T_1}$
are different in $(0010)$ compared to $(1000)$, since the feedback
messages are different. Since $\mathsf{T_1}$ knows its own transmitted
symbol $X_{1i}$, the common message $X_{2i,c}$ needs to be decoded
from $Y_{2i} - g_cX_{1i}$, for which the decoding constraint is
\begin{eqnarray}\label{eq:22mac}
R_{2c}& \leq &\log\left( 1 + \frac{\lambda_{2c}
  \mathsf{SNR}}{\lambda_{2p} \mathsf{INR} + 1}\right).
\end{eqnarray}
Since $\alpha \leq 1$, i.e., $\mathsf{SNR} \geq \mathsf{INR}$, if a
rate $R_{2c}$ satisfies the constraint \eqref{eq:2mac}, then it also
satisfies the constraint \eqref{eq:22mac}. Thus
$\mathcal{R}^{(0010)}$ is achievable if $\mathcal{R}^{(1000)}$ is
achievable, and the proof is complete.\\
\end{proof}

\subsubsection{Achievability for the $\mathbf{(0010)}$ Feedback Model}
\label{sec:(0010)_g}
\paragraph{{\bf Weak Interference}, $\alpha \leq 1$} The corner points $\overline{\mathcal{K}}_{\rm B}$ and $\overline{\mathcal{K}}_{\rm D}$
characterize the outer-bound. From Lemma~\ref{th:crosslink} and
achievability of the $(1000)$ feedback model, we know that
$\mathcal{P}_{\rm D}$, described in Table~\ref{table:rates} is
achievable. The gap of $\mathcal{P}_{\rm D}$ from
$\overline{\mathcal{K}}_{\rm D}$ for $(0010)$ feedback model is
evaluated in Appendix~\ref{sec:kdfromd} and is found to be
$2.59$~bits/Hz. However, to achieve within a constant number of bits
of $\overline{\mathcal{K}}_{\rm B}$, we propose a different achievable
scheme based on block-Markov encoding at ${\sf T}_2$ and dirty paper
encoding at ${\sf T}_1$.

\emph{Encoding:} At ${\sf T}_2$, the message is split into two parts
 $X_{2i,d}$ and $X_{2i,nd}$ with rates $R_{2,d}$ and $R_{2,nd}$, such
 that $R_{2,d} + R_{2,nd} = R_2$. The transmitted message in the
 $i^{\rm th}$ block is
\begin{equation}
  X_{2i} = \sqrt{\lambda_{2,d}}X_{2i-1,d} + \sqrt{\lambda_{2,nd}}X_{2i,nd}  + \sqrt{1 - \lambda_{2,d} - \lambda_{2,nd}}X_{2i,d}
\end{equation}
such that the power constraint is $\lambda_{2,nd} + \lambda_{2,d} \leq
1$. At ${\sf T}_1$, assuming that $X_{2i-1,d}$ can be decoded from the
cross-link feedback, before the $i^{\rm th}$ block of transmission,
the message to be transmitted is encoded into $X_{1i}$ using dirty
paper coding, treating $g_c\sqrt{\lambda_{2,d}}X_{2i-1,d}$ as
interference. The encoded message is denoted as $X_{1i}$, and its rate
is denoted by $R_1$.

\emph{Decoding:} At ${\sf D}_2$, backward decoding is applied. In the
$(i+1)^{\rm th}$ block, $X_{2i,d}$ and $X_{2i+1,nd}$ are assumed to be
decoded. To decode $X_{2i-1,d}$ and $X_{2i,nd}$ from the $i^{\rm th}$
block, $g_d\sqrt{1 - \lambda_{2,d} - \lambda_{2,nd}}X_{2i,d}$ is
subtracted from $Y_{2i}$ and $X_{1i}$ is treated as noise. At ${\sf
  D}_1$, dirty paper decoding is performed to decode the message from
${\sf T}_1$ assuming $X_{2i,nd}$ and $X_{2i,d}$ as noise. At ${\sf
  T}_1$, after the $i^{\rm th}$ transmission block $Y_{2i}$ is
received from the cross-link feedback from ${\sf D}_2$. Since
$X_{2i-1,d}$ is assumed to be known at ${\sf T}_1$ before the $i^{\rm
  th}$ block, from $Y_{2i}$, $g_cX_{1i} +
g_d\sqrt{\lambda_{2,d}}X_{2i-1,d}$ is subtracted to decode $X_{2i,nd}$
and $X_{2i,d}$.

\emph{Choice of power and rate allocation:} Using the above encoding
and decoding strategy, power and rate allocation for the rate pair
described by $\mathcal{P}_{\rm B2}$ in Table~\ref{table:lambdas2} and
\ref{table:rates2} is feasible in the weak interference regime. The
gap of $\mathcal{P}_{\rm B2}$ from
$\overline{\mathcal{K}}_{\rm B}$ is computed in
Appendix~\ref{sec:kbfromb2} and the gap is found to be $4.59$~bits/Hz.

\begin{table}[t] \caption{Power allocation for the $(0010)$ feedback model}
\label{table:lambdas2}
\centering
\begin{tabular}{ |c|l|c|c |c|}
\hline %
Corner Point  & $\alpha$& $\lambda_{1}$ & $\lambda_{2,d}$ &$\lambda_{2,nd}$   \\\hline %
& $[0,1/2)$ & 1 & $1 - 1/{\sf INR}$ & $1/{\sf INR} - {\sf SNR}/{\sf INR}$     \\ \cline{2-2} \cline{4-5}
$\mathcal{P}_{\rm B2}$  & $[1/2, 1]$ &  & $1 - 1/{\sf INR}$ & 0  \\ \hline
\end{tabular}\\
\vspace{0.3cm}
\begin{tabular}{ |c|l|c|c |c|c|}\hline
Corner Point  & $\alpha$& $\lambda_{1p}$ & $\lambda_{2p}$ &$\lambda_{1c}$ & $\lambda_{2c}$\\ \hline
 $\mathcal{P}_{\rm B2}$ & $(1,2)$ & 0 & 0 & $1$ & ${\sf INR}/{\sf SNR^2}$   \\ \hline
 $\mathcal{P}_{\rm D2}$ & $(1,2)$ & 0 & 0 & ${\sf INR}/{\sf SNR^2}$ & $1$   \\ \hline
 $\mathcal{P}_{\rm E}$ & $[2,\infty)$ & 0 & 0 & $1$ & $1$   \\ \hline
\end{tabular}
\end{table}

\begin{table}[t]\caption{Rate allocation for the $(0010)$ feedback model}
\label{table:rates2}
\centering
\begin{tabular}{ |c|l|c|c |c|}
\hline %
Corner Point  & $\alpha$& $R_{1}$ & $R_{2,d}$ &$R_{2,nd}$   \\\hline %
& $[0,1/2)$ & $\log(1 + {\sf SNR}/2)$ & $\log({\sf INR}/{2})$ & $\log({\sf SNR}/{\sf INR^2})$     \\ \cline{2-2} \cline{4-5}
$\mathcal{P}_{\rm B2}$  & $[1/2,1]$ & & $1 - 1/{\sf INR})$ & 0  \\ \hline
\end{tabular}\\
\vspace{0.3cm}
\begin{tabular}{ |c|l|c|c |c|c|}\hline
Corner Point  & $\alpha$& $R_{1p}$ & $R_{2p}$ &$R_{1c}$ & $R_{2c}$\\ \hline
 $\mathcal{P}_{\rm B2}$ & $(1,2)$ & 0 & 0 & $\log({\sf INR})$& $\log(1 + {\sf INR}/{\sf SNR})$   \\ \hline
 $\mathcal{P}_{\rm D2}$ & $(1,2)$ & 0 & 0 & $\log(1 + {\sf INR}/{\sf SNR})$ & $\log({\sf INR})$   \\ \hline
 $\mathcal{P}_{\rm E}$ & $[2,\infty)$ & 0 & 0 & $\log(1 + {\sf SNR})$ & $\log(1 + {\sf SNR})$   \\ \hline
\end{tabular}
\end{table}

\paragraph{{\bf Strong Interference}, $\alpha > 1$}
In this regime of interference, the approximate capacity region can be
achieved without feedback. When $1 < \alpha < 2$, the corner points
$\overline{\mathcal{K}}_{\rm B}$ and $\overline{\mathcal{K}}_{\rm D}$
characterize the outer bound of the capacity region. We have shown the
power and rate allocation for rate pairs described by
$\mathcal{P}_{\rm B2}$ and $\mathcal{P}_{\rm D2}$ in
Table~\ref{table:lambdas2} and \ref{table:rates2}, which can be
achieved without any feedback \cite{etkin}.  The gaps of
$\mathcal{P}_{\rm B2}$ and $\mathcal{P}_{\rm D2}$ from
$\overline{\mathcal{K}}_{\rm B}$ and $\overline{\mathcal{K}}_{\rm D}$
respectively are both computed to be $2$ bits/Hz. For $\alpha \geq 2$,
$\overline{\mathcal{K}}_{\rm E}$ is the only non-trivial corner point
on the outer bound. As the interference is strong enough, it can be
completely decoded. Therefore, the channel is equivalent to two
parallel point to point channels. Thus the rate pair
described by $\mathcal{P}_{\rm E}$ in Tables~\ref{table:lambdas2} and
\ref{table:rates2} is achievable. The gap of $\mathcal{P}_{\rm E}$ from
$\overline{\mathcal{K}}_{\rm E}$ is 0. Thus, the achievability of
$(0010)$ feedback model within $4.59$~bits/Hz is complete.\\

\subsubsection{Achievability for the $\mathbf{(1001)}$, $\mathbf{(1101)}$ and $\mathbf{(1111)}$
  Feedback Models} We will find an achievable rate region for
$(1001)$, which is within a constant number of bits away from the
outer-bound of the $(1111)$ feedback model. For $(1111)$, $(1101)$ and
$(1001)$ feedback models, $\overline{\mathcal{K}}_{\rm B}$ and
$\overline{\mathcal{K}}_{\rm D}$ sufficiently characterize the
outer-bound. Using $(1000)$ feedback model, the rate pair described by
$\mathcal{P}_{\rm D}$ in Table~\ref{table:rates} is achievable, and
thus $\mathcal{P}_{\rm D}$ in Table~\ref{table:rates} is achievable
with $(1111)$, $(1101)$ and $(1001)$ feedback models. The gap of
$\mathcal{P}_{\rm D}$ from $\overline{\mathcal{K}}_{\rm D}$ is
evaluated in Appendix~\ref{para:r2-r1r2_l} and \ref{para:r2-r1r2_h} to
be $2.59$ bits/Hz. From symmetry, $(0001)$ feedback model can achieve
a rate-pair within $2.59$ bits/Hz from $\overline{\mathcal{K}}_{\rm
  B}$. As, $\mathcal{R}^{(1001)} \supseteq \mathcal{R}^{(1000)} $, and
$\mathcal{R}^{(1001)} \supseteq \mathcal{R}^{(0001)}$, $\mathcal{R}^{(1001)}$ contains achievable rate pairs within $2.59$
bits/Hz from both $\overline{\mathcal{K}}_{\rm B}$ and
$\overline{\mathcal{K}}_{\rm D}$. Also, as the following relation holds
\begin{equation}
\mathcal{R}^{(1001)} \subseteq \mathcal{C}^\mathrm{(1001)} \subseteq \mathcal{C}^{(1101)} \subseteq
\mathcal{C}^\mathrm{(1111)}
\end{equation}
and since $\mathcal{R}^{(1001)}$ is within $2.59$ bits/Hz of the outer
bound on the $(1111)$ feedback model, we conclude that
$\mathcal{C}^{(1001)}$ is within $2.59$ bits/Hz of
$\mathcal{C}^{(1111)}$. Since the achievability of
$\mathcal{R}^{(1001)}$ directly follows from the achievability of
$\mathcal{R}^{(1000)}$, the approximate capacity region
characterization of all the feedback models of type
$(1\mathsf{x}\mathsf{x}1)$ is complete.\\

\subsubsection{Achievability for the $\mathbf{(1100)}$, $\mathbf{(1110)}$ and $\mathbf{(1010)}$
  Feedback Models} As more feedback can only increase the capacity region, we have
\begin{equation}
\mathcal{C}^{(1100)} \subseteq \mathcal{C}^{(1110)}.
\end{equation}
The outer bound of the $(1110)$ feedback model is characterized by the
corner points $\overline{\mathcal{K}}_{\rm B}$ and
$\overline{\mathcal{K}}_{\rm D}$. Note that in the previous subsection, it is proved that the rate pair $\mathcal{P}_{\rm D}$, which is described by Table~\ref{table:rates}, is within $2.59$ bits/Hz
from $\overline{\mathcal{K}}_{\rm D}$. As the achievable rate region
$\mathcal{R}^{(1000)}$ contains $\mathcal{P}_{\rm D}$ and
\begin{equation}
\mathcal{R}^{(1000)} \subseteq \mathcal{R}^{(1100)} \subseteq \mathcal{R}^{(1110)} \text{ and } \mathcal{R}^{(1000)} \subseteq \mathcal{R}^{(1010)},
\end{equation}
$(1100)$, $(1110)$ and $(1010)$ feedback models also contain $\mathcal{P}_{\rm D}$, which is within $2.59$ bits/Hz
from $\overline{\mathcal{K}}_{\rm D}$. Now, we show the achievability
of a rate pair within a constant number of bits from
$\overline{\mathcal{K}}_{\rm B}$ for $(1100)$, $(1110)$ and $(1010)$
feedback models.\\

\paragraph{{\bf Weak Interference}, $\alpha \leq 1$}
In this regime, from Lemma~\ref{th:crosslink}, we know that
$\mathcal{R}^{(0100)} \supseteq \mathcal{R}^{(0001)}$. As we have
$\mathcal{R}^{(0110)} \supseteq \mathcal{R}^{(0100)}$, then
$\mathcal{R}^{(0110)} \supseteq \mathcal{R}^{(0100)} \supseteq
\mathcal{R}^{(0001)}$. From symmetry, we know that the achievable rate
region $\mathcal{R}^{(0001)}$ contains a rate pair within
$2.59$~bits/Hz from $\overline{\mathcal{K}}_{\rm B}$. Thus, the rate
regions $\mathcal{R}^{(0100)}$ and $\mathcal{R}^{(0110)}$ and
subsequently $\mathcal{R}^{(1100)}$ and $\mathcal{R}^{(1110)}$ contain
a rate pair within $2.59$ bits/Hz from $\overline{\mathcal{K}}_{\rm
  B}$.

The $(1010)$ feedback model can achieve any rate pair, which the
$(0010)$ feedback model can achieve. Since the rate pair described by
$\mathcal{P}_{\rm B2}$ in Table~\ref{table:rates2} is achievable
within $4.59$~bits/Hz from $\overline{\mathcal{K}}_{\rm B}$ for
$(0010)$ feedback model, thus $\mathcal{P}_{\rm B2}$ is also
achievable with $(1010)$ feedback model and is within $4.59$~bits/Hz
from $\overline{\mathcal{K}}_{\rm B}$.\\

\paragraph{{\bf Strong Interference}, $\alpha > 1$}
In this regime, recall that the achievable rate region
$\mathcal{R}^{(1000)}$ itself contains the rate pair $\mathcal{P}_{\rm
  B}$, described in Table~\ref{table:rates}, which is within $2.59$
bits/Hz from $\overline{\mathcal{K}}_{\rm B}$. Thus,
$\mathcal{R}^{(1100)}$, $\mathcal{R}^{(1010)}$ and
$\mathcal{R}^{(1110)}$ also contain a rate pair within $2.59$ bits/Hz
from $\overline{\mathcal{K}}_{\rm B}$.\\

\subsubsection{Achievability for the $\mathbf{(0110)}$ Feedback Model}
\paragraph{{\bf Weak Interference}, $\alpha \leq 1$}
In this regime of interference, the outer bound on the capacity region
of the $(0110)$ feedback model is characterized by the corner points
$\overline{\mathcal{K}}_{\rm B}$ and $\overline{\mathcal{K}}_{\rm
  D}$. When $\alpha \leq 1$, due to Lemma~\ref{th:crosslink} we know
that $\mathcal{R}^{(1000)} \subseteq \mathcal{R}^{(0010)}$ and
$\mathcal{R}^{(0001)} \subseteq \mathcal{R}^{(0100)}$. Also,
$\mathcal{R}^{(1000)}$ contains the rate pair $\mathcal{P}_{\rm D}$,
described in Table~\ref{table:rates}, which is within $2.59$ bits/Hz of
$\overline{\mathcal{K}}_{\rm D}$ as shown in
Appendix~\ref{sec:kdfromd}. Symmetrically a rate pair is achievable
within $2.59$~bits/Hz from $\overline{\mathcal{K}}_{\rm B}$. Thus,
$\mathcal{R}^{(0010)}$ and $\mathcal{R}^{(0100)}$ also contain rate
pairs within $2.59$ bits/Hz from $\overline{\mathcal{K}}_{\rm D}$ and
$\overline{\mathcal{K}}_{\rm B}$. Consequently $\mathcal{R}^{(0110)}$
includes a rate pair, which is within $2.59$~bits/Hz of both
$\overline{\mathcal{K}}_{\rm D}$ and $\overline{\mathcal{K}}_{\rm
  B}$. \\

\paragraph{{\bf Strong Interference}, $\alpha > 1$}
The rate pair achievable by the $(0010)$ feedback model is also
achievable by the $(0110)$ feedback model. The corner points
$\overline{\mathcal{K}}_{\rm B}$ and $\overline{\mathcal{K}}_{\rm D}$,
which characterize the outer bound of the capacity region when $1 <
\alpha < 2$, are both achievable within a constant number of bits
without feedback by the rate pairs described by $\mathcal{P}_{\rm B2}$ and
$\mathcal{P}_{\rm D2}$ in Table~\ref{table:rates2}. The gap of $\mathcal{P}_{\rm D2}$ from $\overline{\mathcal{K}}_{\rm D}$ is
computed in Appendix~\ref{sec:kdnfb} and is found to be $2$~bits/Hz. Due to
symmetry, the gap of $\mathcal{P}_{\rm B2}$ from
$\overline{\mathcal{K}}_{\rm B}$ is also $2$~bits/Hz. When $\alpha >
2$, the only non-trivial corner point on the outer bound is
$\overline{\mathcal{K}}_{\rm E}$, which is achievable without any
feedback with 0 gap from the rate pair $\mathcal{P}_{\rm E}$ described
in Table~\ref{table:rates2}.\\


\subsubsection{Sum-capacity of all feedback models}
In order to characterize the sum-capacity of all feedback models, with
at least one feedback link, in the weak interference regime, we use
the outer bound on the $(1111)$ feedback model. In the weak
interference regime, the corner point $\overline{\mathcal{K}}_{\rm D}$
for $(1111)$ feedback model is outside the capacity region of all
feedback models. We know that the corner point $\mathcal{P}_{\rm D}$,
described by Table~\ref{table:rates} is achievable for all feedback
models in the weak interference regime. The distance of the sum-rate
described by $\mathcal{P}_{\rm D}$ from $\overline{\mathcal{K}}_{\rm
  D}$ for the $(1111)$ feedback model can be computed from Appendix
\ref{para:r2-r1r2_l} and is found to be $4.59$~bits/Hz. Thus, all
feedback models can achieve a rate pair within $4.59$~bits/Hz from
$\overline{\mathcal{K}}_{\rm D}$. Since $\overline{\mathcal{K}}_{\rm
  D}$ for the $(1111)$ feedback model lies on the sum-rate outer bound on the
$(1111)$ feedback models, thus it lies on the sum-rate outer bound of
all feedback models. For all feedback models, we have shown an
achievable rate pair, $\mathcal{P}_{\rm D}$, which is within
$4.59$~bits/Hz from the sum-rate outer bound of $(1111)$ feedback
model. Thus, the sum-capacity of all feedback models, in the weak
interference regime, is within $4.59$~bits/Hz of each other. \\

\subsubsection{Sum-capacity of $(1\mathsf{xxx})$ feedback models}
In the strong interference regime, all feedback models of type
$(1\mathsf{xxx})$, can achieve a rate pair described by
$\mathcal{P}_{\rm D}$, which is within $3$~bits/Hz from the corner
point $\overline{\mathcal{K}}_{\rm D}$ that lies on sum-rate outer
bound of the $(1111)$ feedback model as shown in
Appendix~\ref{para:r2-r1r2_h}. Thus, in strong interference regime,
all feedback models of type $(1\mathsf{xxx})$, in the strong
interference regime, is within $3$~bits/Hz of each other.


\section{Conclusion}
\label{sec:conclude}
In this paper, we characterize the capacity region of all channel
output feedback models in a two user symmetric interference
channel. Depending on whether an infinite capacity feedback link
exists between a receiver and a transmitter, a total of 9 canonical
feedback models are present. In case of the symmetric linear
deterministic interference channel, we find the exact capacity region,
while for the Gaussian channel we find the approximate capacity region
within at most $4.59$~bits/Hz for all the 9 feedback
models. Interestingly, in the weak interference regime all models of
feedback have the identical capacity region except the feedback model
with a single direct feedback link. In other words, all feedback
models (except the single direct link feedback model) have the same
capacity region as the capacity region achievable with all four
feedback links. In particular, this includes that the capacity region
of the single cross link feedback model is identical with the capacity
region of the feedback model with all four feedback links. Although
the single direct-link feedback has a smaller capacity region than
other feedback models, in the weak interference regime, its
sum-capacity is identical to the sum-capacity of the rest of the
feedback models. In the strong interference regime as well, single
direct-link feedback is sufficient to achieve the same sum-capacity as
that achievable with all four feedback links.

To prove these results, we proposed two new outer-bounds, one for the single direct link
feedback model and another for the feedback model with all four
feedback links. The two new outer bounds together with the cut-set
bound form a comprehensive outer bound for all feedback models, which allow for exact/approximate capacity region calculations for deterministic/Gaussian channel models. 

In the weak interference regime, two new achievable strategies are
proposed: one which is based on Han-Kobayashi type message splitting
and the other which is based on block-Markov coding (at one
transmitter) and dirty paper coding (at the other
transmitter). Together, the two strategies achieve the
exact/approximate capacity region for all 9 canonical feedback models
for deterministic/Gaussian channels. In the achievable strategy
involving Han-Kobayashi type message splitting, the transmitted
message from each of the transmitters is split into two parts: private
and common. The common part of the message of one of the transmitters
is transmitted twice: once by the transmitter, which generates it, and
once again (in the subsequent block) by the other transmitter after
decoding it. The rate of the common message, which is re-transmitted is
finely tuned so that it is decodable at the intended receiver after
its first transmission, while it is decodable at the interfering
receiver only after its second transmission. Although the common
message first causes interference at one of the receivers, it allows
for higher communication rates after interference resolution in the
subsequent block. In the achievable strategy involving block-Markov
encoding and dirty paper encoding/decoding, one of the transmitters
employs block-Markov encoding, thereby correlating the interference it
generates over blocks. The other transmitter knows the channel output
via feedback, and using the knowledge of correlation of interference,
it encodes its message using dirty paper coding to make its intended
signal robust against future interference.

In the strong interference regime, feedback helps create a relay
route, which is better than the direct channel from a transmitter to
its intended receiver. The messages generated at a transmitter are
first passed on to the interfering receiver. The interfering receiver
then passes it on to its own transmitter (via feedback), which can
then relay it to the intended receiver. This way the intended receiver
receives the message through an alternate path. Since the interference
is stronger than the direct channel, relaying of messages can support
higher rates than otherwise.


\bibliographystyle{IEEEtran}
\bibliography{references}

\appendix

\subsection{Proof of Lemma~\ref{th:sumcap2}}
\label{apd:2}
Let $V_{1i}=\mathbf{S}^{q-m}X_{1i}$ and
$V_{2i}=\mathbf{S}^{q-m}X_{2i}$. We know that $X_{1i}$ and $X_{2i}$
are given by
\begin{eqnarray*}
X_{1i}  =  f_{1i}(W_1, Y_1^{i-1}), \text{ } X_{2i} = f_{2i}(W_2)
\end{eqnarray*}
where $f_{1i}(.), f_{2i}(.)$ are some deterministic functions. We have
\begin{eqnarray}
\lefteqn{N(2R_1 + R_2)} \nonumber \\ & \le & 2H(W_1) + H(W_2) \nonumber \\
 & \stackrel{\text{(a)}}{=} & H(W_1) + H(W_1|W_2) + H(W_2) \nonumber \\
 & \stackrel{\mathrm{(Fano)}}{\leq} & I(W_1; Y_1^{N}) + I(W_1; Y_1^N|
W_2) + I(W_2; Y_2^N) + N(\epsilon_{1N} + \epsilon_{2N} + \epsilon_{3N}) \nonumber \\
& = & H(Y_1^N) - H(Y_1^N|W_1) + H(Y_1^N|W_2) - H(Y_1^N| W_1 W_2) + \nonumber  \\
& & H(Y_2^N) - H(Y_2^N|W_2) + N\epsilon_N, \label{eq:proof21}
\end{eqnarray}
where $\epsilon_{1N}$, $\epsilon_{2N}$ and $\epsilon_{3N}$ correspond
to the Fano's inequality applied to three different entropy terms, and
$\epsilon_N = 3\max(\epsilon_{1N}, \epsilon_{2N}, \epsilon_{3N})$ and
(a) holds because $W_1$ and $W_2$ are
independent. Rearranging~\eqref{eq:proof21} yields
\begin{eqnarray*}
\lefteqn{N(2R_1 + R_2)} \nonumber \\
& \leq & H(Y_1^N) + \underbrace {H(Y_2^N) - H(Y_1^N| W_1)} +\underbrace{
H(Y_1^N|W_2) - H(Y_1^N| W_1 W_2) - H(Y_2^N| W_2)} + N\epsilon_N \nonumber \\
& \stackrel{\text{(b)}}{\leq} & H(Y_1^N) + \underbrace{H(Y_2^N) + H(V_2^N| Y_2^N) - H(Y_1^N|W_1)}
+ \underbrace{ H(Y_1^N|W_2) - H(Y_2^N| W_2)} + N\epsilon_N,
\end{eqnarray*}
where (b) is true as entropy for discrete random variables is always
positive.  The three sub-expressions are independently bounded. The
first sub-expression satisfies $H(Y_1^N) = \sum_{i = 1}^N
H(Y_{1i}|Y_{1}^{i-1}) \leq \sum_{i = 1}^N H(Y_{1i})$ due to the chain
rule of entropy followed by the fact that removing conditioning does
not reduce entropy. The second sub-expression is bounded as follows:
\begin{equation}
H(Y_2^N) + H(V_2^N|Y_2^N) - H(Y_1^N|W_1) = H(Y_2^N|V_2^N) + H(V_2^N) - H(Y_1^N|W_1) \label{eq:subexpbounding1}
\end{equation}

Observe the following:
\begin{eqnarray}
\lefteqn{H(V_2^N) - H(Y_1^N| W_1)} && \nonumber \\
& \stackrel{\text{(c)}}{=} & H(V_2^N|W_1) - \sum_{i = 1}^N H(Y_{1i}| W_1, Y_1^{i-1}) \nonumber \\
& \stackrel{\text{(d)}}{=} & H(V_2^N|W_1) - \sum_{i = 1}^N H(Y_{1i}| W_1, Y_1^{i -1}, X_1^{i}) \nonumber \\
& \stackrel{\text{(e)}}{=} & H(V_2^N|W_1) - \sum_{i = 1}^N H(V_{2i}| W_1, V_2^{i-1}, Y_1^{i-1}, X_1^{i}) \nonumber \\
 & \stackrel{\text{(f)}}{=}  & \sum_{i=1}^N H(V_{2i}|W_1, V_2^{i-1}) - \sum_{i = 1}^N H(V_{2i}|W_1, V_2^{i-1},X_1^{i}, Y_1^{i-1}) \nonumber \\
 & = &
 \sum_{i=1}^N I(V_{2i}; X_1^{i}, Y_1^{i-1}|W_1, V_2^{i-1}) \nonumber \\ & = &
 \sum_{i=1}^N [H(X_1^{i}|W_1,V_2^{i-1}) + H(Y_1^{i-1}|W_1,V_2^{i-1}, X_1^i)] - [H(X_1^{i}
 |W_1,V_2^i) + H( Y_1^{i-1}|W_1,V_2^i, X_1^i) ] \nonumber \\
 &\stackrel{\text{(g)}}{=} & \sum_{i=1}^N [ H(X_1^{i}|W_1,V_2^{i-1}) - H(X_1^{i}
 |W_1,V_2^i)] \nonumber \\
 & {=} & \sum_{i =1}^N I(X_1^{i}; V_{2i}|W_1, V_2^{i-1}) \nonumber \\
& \stackrel{\text{(h)}}{=}  & \sum_{i =1}^N I(f(V_2^{i-1}, W_1); V_{2i}|W_1,V_2^{i-1}) \nonumber \\
& = &  0 \label{eq:subexpbounding2},
 \end{eqnarray}
where (c) is true because $V_2^N$ depends only on $W_2$ and thus
independent of $W_1$, (d) holds because $X_1^i$ is a deterministic
function of $W_1$ and $Y_1^{i-1}$, (e) is justified because $Y_{1i} =
X_{1i} + V_{2i}$, (f) is due to the chain rule of entropy, (g) holds
because $Y_1^{i-1}$ is a deterministic function of $X_1^{i-1}$ and
$V_2^{i-1}$, (h) is true because of the chain rule because of the
following: $X_{1i}$ depends on $W_1$ and $Y_1^{i-1}$, but $Y_{1i} =
X_{1i-1} + V_{2i-1}$. Thus $X_{1i}$ is function of $W_1$, $V_{2i-1}$,
and $Y_1^{i-2}$. This implies that $X_1^i$ is a function of $W_1$ and
$V_2^{i-1}$ only. Combining \eqref{eq:subexpbounding1} and
\eqref{eq:subexpbounding2}, we have
\begin{eqnarray}
 H(Y_2^N) + H(V_2^N|Y_2^N) - H(Y_1^N|W_1)  =  H(Y_2^N|V_2^N)  =  \sum_{i = 1}^N H(Y_{2i}|V_{2i},Y_{2}^{i-1},V_{2}^{i-1})  \leq  \sum_{i = 1}^N H(Y_{2i}|V_{2i}),
\end{eqnarray}
where the inequality follows from the fact that removing conditioning
cannot decrease entropy.

Finally, for the third sub-expression, we have
\begin{eqnarray}
\lefteqn{H(Y_1^N|W_2) - H(Y_2^N|W_2)} \nonumber  \\
& = & H(Y_1^N|W_2) +H(Y_1^N|Y_2^N,W_2) - H(Y_1^N,Y_2^N|W_2)\nonumber \\
& = & H(Y_1^N|Y_2^N, W_2) - H(Y_2^N|Y_1^N,W_2)  \nonumber \\
& \leq & H(Y_1^N|Y_2^N,W_2) \nonumber  \\
& \stackrel{\text{(j)}}{=} & \sum_{i=1}^N H(Y_{1i}| Y_{2}^N, Y_{1}^{i-1}, W_2) \nonumber \\
& \stackrel{\text{(k)}}{=} & \sum_{i=1}^N H(Y_{1i}| Y_{2}^N, Y_{1}^{i-1}, W_2, X_{2}^i, V_2^i, V_1^{i})  \nonumber \\
& \stackrel{\text{(l)}}{\leq} & \sum_{i=1}^N H(Y_{1i}| V_{1i}, V_{2i})  \label{eq:th31_2}
\end{eqnarray}
 (j) follows from the chain rule of entropy, (k) follows from the
observation that $X_{2}^i$ is a function of only ($W_2, Y_{1}^{i-1},
Y_{2}^{i}$), $V_2^{i}$ is function of $X_2^i$, and $V_1^{i}$ is a
function of ($X_2^{i}$, $Y_{2}^i$), and (l) follows since conditioning
reduces entropy.

Now combining all the expressions together we have
\begin{equation}
N(2R_1 + R_2) \le \sum_{i = 1}^N (H(Y_{1i}) + H(Y_{2i}| V_{2i}) +
H(Y_{1i} | V_{1i},V_{2i}) + N\epsilon_N
\end{equation}
By randomization of time indices and letting $\epsilon_N \to 0$ as $N
\to \infty$, we get

\begin{equation}
2R_1 + R_2 \le H(Y_1) + H(Y_2|V_2) + H(Y_1|V_1, V_2).
\end{equation}
 The RHS is maximized when $X_1$ and $X_2$ are drawn from an
 i.i.d. distribution over $\mathbb{F}_2^q$, where each entry of
 the $q$-bit vector is i.i.d. $\mathrm{Bern}(\frac{1}{2})$. This gives
 us the outer bound as in the statement of Lemma~\ref{th:sumcap2}.

\subsection{Achievable strategy for the corner points of the capacity region of the $(1000)$ feedback model}
\label{apd_a}
 \emph{Encoding:} At the $u^{\rm th}$ transmitter ${\sf T}_u$, in the
 $i^{\rm th }$ block two i.i.d.~bit vectors $X_{ui,c}$ and $X_{ui,p}$
 are generated. The total number of transmission blocks is $B$. Let
 $\mathbf{0}_l = [0,0,\ldots 0]$ such that $|\mathbf{0}_l| = l$. The
 encoding of messages is for all the $B$ blocks is shown in
 Table~\ref{table:encoding}.

\begin{table}[t]
\centering
  \caption{Encoding of messages in the weak interference regime for the $(1000)$ feedback model}
  \begin{tabular}{ | c | c | c | c |}
    \hline
    & Block 1 & Block $i$ & Block $B$   \\ \hline
    Message $X_{1i}$ at ${\sf T}_1$ & $\mathbf{0}_n^T$ & $[X_{1i,c}^T,X_{2i-1,c}^T, X_{1i,p}^T]^T$ & $[X_{1B,c}^T, X_{2B-1,c}^T, X_{1B,p}^T]^T$ \\ \hline
    Message $X_{2i}$ at ${\sf T}_2$ & $[X_{21,c}^T, \mathbf{0}_{l}^T, X_{21,p}^T]^T$ & $[X_{2i,c}^T, \mathbf{0}_{l}^T,X_{2i,p}^T]^T$ & $\mathbf{0}_n^T$ \\
    \hline
  \end{tabular}
      \label{table:encoding}
\end{table}
\subsubsection{Weak interference regime $n \geq m$} In the weak interference regime, we note that $|X_{1i}| = |X_{2i}| = n$. The encoding scheme is
complete, if the cardinality of $X_{1i,c}, X_{2i,c}, X_{1i,p}$ and
$X_{2i,p}$ are specified.

\emph{Decoding:} To allow reliable decoding, we specify the
cardinality of the common and private message for corner points
$\mathcal{K}_{\rm A}$, $\mathcal{K}_{\rm B}$, $\mathcal{K}_{\rm C}$,
$\mathcal{K}_{\rm D}$ \eqref{eq:definecorner} as the following:

\paragraph{Corner point $\mathcal{K}_{\rm A}$}
When $m < \frac{n}{2}$, the
desired corner point is $(n, n- 2m)$, and it is achievable without
feedback~\cite{bresler2}. When $\frac{n}{2} \leq m \leq n$, the
desired corner point is $(n,0)$, which is trivially achievable without
feedback.

\paragraph{Corner point $\mathcal{K}_{\rm C}$}
When ${m < \frac{n}{2}}$, the intersection is at the corner point
$(n-m, n)$, which is identical to the corner point $\mathcal{K}_{\rm D}$
that will be shown to achievable in Appendix~\ref{sec:kd}. When $\frac{2n}{3} \leq m \leq
n$, we conclude from Lemma~\ref{th:sumcap} and Theorem~2.1 that
the corner point is achievable without feedback. When ${\frac{n}{2}
  \le m < \frac{2n}{3}}$, the corner point $(m, 2n-2m)$ can be
achieved as $B\to \infty$, if
\begin{equation}
   |X_{1i,c}| = 2m - n, \text{ } |X_{2i,c}| = n - m, \text{ } |X_{1i,p}| = |X_{2i,p}| = n-m, \text{ } \mathbf{0}_l = l = 2m -n.
\end{equation}

\paragraph{Corner point $\mathcal{K}_{\rm D}$}
\label{sec:kd}
The corner point $(n-m,n)$ can be achieved as $B \to \infty $ if
\begin{equation}
   |X_{1i,c}| =  0, \text{ } |X_{2i,c}| = m, \text{ } |X_{1i,p}| = |X_{2i,p}| = n-m.
\end{equation}

${\sf D}_1$ and ${\sf D}_2$ respectively perform backward and forward
decoding. Due to backward decoding at ${\sf D}_1$, before decoding the
$i ^{\rm th}$ block $X_{2i,c}$ is known. Thus, $X_{2i,c}$ can be
subtracted from $Y_{1i}$, after which $X_{1i,c}, X_{2i-1,c}$ and
$X_{1i,p}$ can be decoded. Due to forward decoding at ${\sf D}_2$,
while decoding the $i^{\rm th}$ block $X_{2i-1,c}$ is already
known. Thus, $X_{2i-1,c}$ can be subtracted from $Y_{2i}$, after which
$X_{2i,c}, X_{1i,c}$ and $X_{2i,p}$ can be decoded.

\subsubsection{Strong interference regime $n < m$}
In the strong interference regime, $|X_{1i}| = |X_{2i}| = m$. Proving
achievability for the following two corner points is sufficient to show
the achievability of the outer-bound.

\paragraph{Corner point $\mathcal{K}_{\rm B}$}
The desired corner point is $(n, m -n)$, which is achievable without
feedback for $n < m \le 2n$~\cite{bresler2}. For $m>2n$, the rate pair
$(n, m -n)$ can be achieved if
\begin{equation}
  |X_{1i,c}| = n \text{ } |X_{2i,r}| = m - n, \text{ } |X_{1i,p}| = |X_{2i,p}| = 0, \text{ } \mathbf{0}_l = l = n.
\end{equation}

\paragraph{Corner point $\mathcal{K}_{\rm D}$}
In this case, the desired corner point is $(0,m)$. It can be achieved
if
\begin{equation}
 |X_{1i,c}| =  0, \text{ } |X_{2i,c}| = m, \text{ } |X_{1i,p}| = |X_{2i,p}| = 0.
\end{equation}
In this case, the unit ${\sf D}_1$-feedback-${\sf T}_1$ entirely
serves as a relay node. Forward decoding at both receivers is used to decode the desired messages.

\subsection{Proof of Theorem~\ref{th:sum-capg}}
\label{pthscg}
Let's define $S_{1i} = g_cX_{1i} +
Z_{2i}$ and $S_{2i} = g_cX_{2i} + Z_{1i}$
 \begin{eqnarray}\label{eq16}
\lefteqn{N(R_1 + R_2)} \nonumber \\
 & \leq & H(W_1, W_2) = H(W_1|W_2) + H(W_2) \label{eq_et} \\ &
 \stackrel{\mathrm{(Fano)}}{\leq} & I(W_1; Y_1^N| W_2) + I(W_2;Y_2^N)
 + N(\epsilon_{1N}+ \epsilon_{2N})\nonumber  \\
 & = & \underbrace{ h(Y_1^N|W_2) - h(Y_1^N| W_1, W_2) - h(Y_2^N|W_2)}
 + h(Y_2^N) + N\epsilon_N \label{equb}
\end{eqnarray}
where $\epsilon_{1N}, \epsilon_{2N}$ appear after applying Fano's inequality to the two entropy
terms in \eqref{eq_et}. Also, $\epsilon_N = \epsilon_{1N}+ \epsilon_{2N}$.
We now bound the expression in the under-brace in \eqref{equb} as
\begin{eqnarray}
\lefteqn{h(Y_1^N|W_2) - h(Y_1^N| W_1, W_2) - h(Y_2^N|W_2)} \nonumber \\
 & \stackrel{\text{(a)}}{\leq} & h(Y_1^N| W_2) + h(Y_1^N|Y_2^N, W_2) -
 h(Y_1^N| Y_2^N, W_2) - h(Y_2^N| W_2) - \sum_{j = 1}^N h(Z_{1i}) \nonumber \\
 & = & h(Y_1^N|W_2) + h(Y_1^N|Y_2^N, W_2) - h(Y_1^N, Y_2^N| W_2) - \sum_{i = 1}^N h(Z_{1i}) \nonumber \\
 & = & h(Y_1^N|Y_2^N, W_2) - h(Y_2^N|Y_1^N, W_2) - \sum_{i = 1}^N h(Z_{1i}) \nonumber  \\
 & \stackrel{\text{(b)}}{\leq} & h(Y_1^N| Y_2^N, W_2) - \sum_{i = 1}^N
 [h(Z_{1i}) + h(Z_{2i}) ] \nonumber \\
 & \stackrel{\text{(c)}}{=} & \sum_{ i = 1}^N  h(Y_{1i}| Y_2^N, W_2, Y_1^{i -1}) - \sum_{i = 1}^N [h(Z_{1i}) + h(Z_{2i}) ]  \nonumber
 \end{eqnarray}
 \begin{eqnarray}
 & \stackrel{\text{(d)}}{=} & \sum_{i=1}^N h(Y_{1i}| Y_2^N, W_2, Y_1^{i-1},
 X_2^i, S_{1i}) - \sum_{i = 1}^N [h(Z_{1i}) + h(Z_{2i}) ] \nonumber \\
 & \stackrel{\text{(e)}}{\leq} & \sum_{i = 1}^N h(Y_{1i}| X_{2i},
 S_{1i}) - \sum_{i = 1}^N [h(Z_{1i}) + h(Z_{2i}) ] \label{eq29}
 \end{eqnarray}
where (a) holds since
\begin{eqnarray*}
h(Y_1^N| W_1, W_2) & = & \sum_{i=1}^N h(Y_{1i} | W_1, W_2, Y_1^{i-1} )
\geq \sum_{i=1}^N h(Y_{1i} | W_1, W_2, Y_1^{i-1}, Y_2^{i-1} ) \\ & = &
\sum_{i=1}^N h(Y_{1i} | W_1, W_2, Y_1^{i-1}, Y_2^{i-1} , X_{1i},
X_{2i}) = \sum_{i= 1}^N h(Z_{1i}),
\end{eqnarray*}
(b) follows from\begin{eqnarray*}
h(Y_2^N|Y_1^N, W_2) & = & \sum_{i= 1}^N h(Y_{2i}| Y_1^N, W_2,
Y_2^{i-1}) = \sum_{i= 1}^N h(Y_{2i}| Y_1^N, W_2, Y_2^{i-1}, X_{2i})\\
 & \geq &  \sum_{i= 1}^N h(Y_{2i}| Y_1^N, W_2, Y_2^{i-1}, X_{2i}, X_{1i})  =  \sum_{i = 1}^N h(Z_{2i}) ,
\end{eqnarray*}
 (c) is due to the chain rule of entropy, (d) holds because given $W_2$
and $X_2^N$ can be precisely determined and $Y_{2i} = X_{2i} + S_{1i}$
and thus given $Y_{2i}$ and $X_{2i}$, $S_{2i}$ can be precisely
determined, (e) uses the fact that removing conditioning does not
increase the entropy.

We plug-in this part in the original sum-rate bound \eqref{equb} to get
\begin{equation}
R_1 + R_2 \leq \frac{1}{N} \left(h(Y_2^N) + \sum_{i = 1}^N h(Y_{1i}| X_{2i}, S_{1i}) - \sum_{i = 1}^N [h(Z_{1i} + h(Z_{2i})]  \right) + \epsilon_N \nonumber
\end{equation}

Letting $N \to \infty$ we can make $\epsilon_N \to 0$. Moreover
applying the chain rule of entropy and noting that removing
conditioning does not increase entropy, the following outer bound is
obtained
\begin{equation}
R_1 + R_2 \leq \frac{1}{N} \left( \sum_{i = 1}^N h(Y_{2i}) + \sum_{i =
1}^N h(Y_{1i}| X_{2i}, S_{1i}) - \sum_{i = 1}^N [h(Z_{1i}) +
h(Z_{2i})] \right) \nonumber
\end{equation}

By simply interchanging the indices of the users,
i.e., following the substitution $1 \to 2$ and vice versa, we obtain
\begin{equation}\label{eq:finalg}
R_1 + R_2 \leq \frac{1}{N} \left(\sum_{i =1}^N h(Y_{1i}) + \sum_{i =
1}^N h(Y_{2i}| X_{1i}, S_{2i}) - \sum_{i = 1}^N [h(Z_{1i}) +
h(Z_{2i})] \right)
\end{equation}

Assuming that both $X_1$ and $X_2$ is drawn from complex Gaussian
distributions with mean $0$ and variance $1$, and the correlation
between $X_1$ and $X_2$ is $\rho$, i.e.~$\rho =
\mathsf{E}[X_1{X_2}^*]$, and then the \eqref{eq:finalg} can be
expressed in terms of ${\sf SNR}$ and ${\sf INR}$ as
\begin{equation}
R_1 + R_2 \leq \sup_{0 \leq |\rho| \leq 1} \left\{ \log \left( 1 +
\frac{(1 - |\rho|^2) \mathsf{SNR}}{1 + (1 - |\rho|^2)
\mathsf{INR}}\right) + \log (1 + \mathsf{SNR + INR} +
2|\rho|\sqrt{\mathsf{SNR. INR}} )\right\}.
\end{equation}
which is the statement of Theorem~\ref{th:sum-capg}.

\subsection{Proof of Theorem~\ref{th:sum-capg2}}
\label{proof_th_sum-capg2}
\begin{eqnarray}
\lefteqn{N(2R_1 + R_2)} \nonumber \\
& \stackrel{\text{(a)}}{=} &  H(W_1) + H(W_1|W_2) + H(W_2) \nonumber  \\
 & \stackrel{\text{(b)}}{\leq} & I(W_1; Y_1^{N}) + I(W_1; Y_1^N|
W_2) + I(W_2; Y_2^N) + N(\epsilon_{1N} + \epsilon_{2N} + \epsilon_{3N}) \nonumber  \\
& = & h(Y_1^N) - h(Y_1^N| W_1) + h(Y_1^N|W_2) - h(Y_1^N| W_1 W_2) + h(Y_2^N)  - h(Y_2^N|W_2) + N\epsilon_N
\end{eqnarray}
(a) is due to the independence of the messages at the two transmitters. (b) follows due to applying Fano's inequality to each of the entropy terms and $\epsilon_N =
3\max(\epsilon_{1N}, \epsilon_{2N}, \epsilon_{3N})$. Rearranging the
terms, the following expression is obtained
\begin{eqnarray}\label{eq:3subex}
  h(Y_1^N) + \underbrace {h(Y_2^N) - h(Y_1^N| W_1)} +\underbrace{
 h(Y_1^N|W_2) - h(Y_1^N| W_1 W_2) - h(Y_2^N| W_2)} + N\epsilon_N
\end{eqnarray}

The three sub-expressions are separately bounded. The first
sub-expression is $h(Y_1^N) = \sum_{i=1}^N h(Y_{1i}|Y_1^{i-1}) \leq
\sum_{i=1}^N h(Y_{1i})$, because removing conditioning does not reduce
entropy.

In order to bound the second sub-expression, observe the following:
\begin{eqnarray}
\lefteqn{ h(S_2^N) - h(Y_1^N|W_1)} \nonumber \\
 & = &  h(S_2^N) - \sum_{i = 1}^N
h(Y_{1i}| W_1 Y_1^{i - 1}) \nonumber \\
 & \stackrel{\mathrm{(c)}}{=} & h(S_2^N)
- \sum_{i=1}^N h(Y_{1i}| W_1 Y_1^{i-1} X_1^{i}) \nonumber \\
 & \stackrel{\mathrm{(d)}}{=} & h(S_2^N| W_1) - \sum_{i = 1}^N h(S_{2i}|
W_1 Y_1^{i-1} X_1^i) \nonumber \\
& \stackrel{\mathrm{(e)}}{=}& \sum_{i = 1}^N
h(S_{2i}| W_1 S_2^{i-1}) - \sum_{i = 1}^N h(S_{2i}| W_1 Y_1^{i-1}
X_1^i) \nonumber \\
 & \stackrel{\mathrm{(f)}}{=}& \sum_{i = 1}^N h(S_{2i}| W_1
S_2^{i-1}) - \sum_{i = 1}^N h(S_{2i}| W_1 Y_1^{i-1} X_1^i S_2^{i-1}) \nonumber \\
 & = & \sum_{i = 1}^N I(S_{2i}; Y_1^{i-1} X_1^i | S_2^{i-1} W_1 ) \nonumber \\
& = & \sum_{i = 1}^N h(X_1^i Y_1^{i-1} | S_2^{i-1} W_1) - h(X_1^i
Y_1^{i-1}| S_2^i W_1 ) \nonumber \\
 & = & \sum_{i = 1}^N h(X_1^i | S_2^{i-1}
W_1) + h(Y_1^{i-1} | X_1^{i} S_2^{i-1} W_1) - ( h(X_1^i| S_2^i W_1 ) +
h( Y_1^{i-1}| X_1^{i} S_2^i W_1 ) ) \nonumber \\
 &\stackrel{\mathrm{(g)}}{=} & \sum_{i =
1}^N h(X_1^i | S_2^{i-1} W_1) - \sum_{i =1}^N h(X_1^i| S_2^i W_1 ) \nonumber \\
 & \stackrel{\mathrm{(h)}}{=} &\sum_{i = 1}^N \sum_{j = 1}^i h(X_{1j}| X_1^{j-1}, S_2^{i-1}, W_1) - \sum_{i = 1}^N \sum_{j = 1}^i  h(X_{1j}| X_1^{j-1}, S_2^{i}, W_1) \nonumber \\
& \stackrel{\mathrm{(i)}}{=}&  \sum_{i = 1}^N \sum_{j = 1}^i h(X_{1j}| X_1^{j-1}, S_2^{i-1}, W_1, Y_1^{j-1}) - \sum_{i = 1}^N \sum_{j = 1}^i  h(X_{1j}| X_1^{j-1}, S_2^{i}, W_1, Y_1^{j-1}) \nonumber \\
&\stackrel{\mathrm{(j)}}{=}& 0 \label{eq:s2}
\end{eqnarray}
(c) holds because $X_1^i$ is a deterministic function of $W_1$ and
$Y_1^{i-1}$, (d) is justified because the message $W_1$ is independent
of $W_2$, and $S_{2i}$ depends only on $W_2$ and the noise
$Z_{1i}$, (e) holds due to the chain rule of entropy (f) is because
$S_{2}^{i-1}$ can be precisely determined from $X_{1}^i$ and
$Y_{1}^{i-1}$ (g) holds because given $X_1^{i-1}$ and $S_2^{i-1}$,
$Y_1^{i-1}$ can be precisely determined, (h) is obtained by applying the chain
rule of entropy to both of the summation terms, (i) holds as $Y_1^{j-1}$
can be precisely determined using $X_1^{j-1}$ and $S_2^{j-1}$, (j) is
true because given $W_1$ and $Y_1^{j-1}$, $X_{1j}$ can be precisely
determined and hence the value of each of the entropy terms is 0.

Let us introduce $S_{2i}' = g_cX_{2i} + Z_{2i}'$, where for every $i$,
the $Z_{2i}'$ is independently distributed with
$\mathcal{CN}(0,1)$. Since entropy is a function of the probability
density function, $h(S_2^N) = h(S_2'^N)$. From \eqref{eq:s2}, we know
that $h(Y_1^N|W_1) = h(S_2^N)$. Thus, $h(Y_1^N|W_1) = h(S_2^N) =
h(S_2'^N)$, which can be used in the second subexpression in
\eqref{eq:3subex} such that
\begin{eqnarray*}
\lefteqn{h(Y_2^N) - h(Y_1^N|W_1)} \\ & = & h(Y_2^N) - h(S_2'^N) \\ & = &
h(Y_2^N ) + h( S_2'^N|Y_2^N) - h(S_2'^N | Y_2^N) - h(S_2'^N) \\ & = &
h(Y_2^N, S_2'^N) - h(S_2'^N ) - h(S_2'^N |Y_2^N) \\ &
\stackrel{\text{(a)}}{\leq} & h(Y_2^N| S_2'^N) - h(S_2'^N |Y_2^N,
X_2^N) \\ & \stackrel{\text{(b)}}{=} & h(Y_2^N| S_2'^N) - h(Z_2'^N
|Y_2^N, X_2^N) \\ & \stackrel{\text{(c)}}{=} & h(Y_2^N| S_2'^N) -
h(Z_2'^N) \\ & \stackrel{\text{(d)}}{\leq} &\sum_{i = 1}^N
(h(Y_{2i}|S'_{2i}) - h(Z'_{2i} ))
\end{eqnarray*}
where (a) holds because conditioning reduces entropy, (b) holds
because $S_2'^N$ is a function of $X_2^N$ and $Z_{2}'^N$, (c) holds
because $Z_2'^N$ is independent of $(Y_2^N, X_2^N)$, (d) holds because
entropy can only increase if conditioning is removed and noise
$Z'_{2i}$ is independent of $Z_{2j}'$ for $i \neq j$.

The third subexpression in \eqref{eq:3subex} is also
bounded with (\ref{eq29}). Putting them together, we finally have
the following bound
\begin{eqnarray*}
N(2R_1 + R_2) & \leq & \sum_{i = 1}^N (h(Y_{1i}) + h(Y_{1i}| S_{1i}
X_{2i}) - h(Z_{1i}) -h(Z_{2i}) +  h(Y_{2i}| S_{2i}') - h(Z_{2i}'))  + N\epsilon_N.
\end{eqnarray*}
Again letting $N \to \infty$ we can make $\epsilon_N \to 0$ and thus
we have the upper bound

\begin{equation} \label{eq:finalg2}
2R_1 + R_2  \leq \frac{1}{N} \left( \sum_{i = 1}^N [h(Y_{1i}) + h(Y_{1i}| S_{1i}
X_{2i}) + h(Y_{2i}| S'_{2i}) - h(Z_{1i}) -h(Z_{2i}) - h(Z'_{2i}) ]     \right).
\end{equation}
Assuming that both $X_1$ and $X_2$ is drawn from complex Gaussian
distributions with mean $0$ and variance $1$, and the correlation
between $X_1$ and $X_2$ is $\rho$, i.e.~$\rho =
\mathsf{E}[X_1{X_2}^*]$, and then the \eqref{eq:finalg2} can be
expressed in terms of ${\sf SNR}$ and ${\sf INR}$ as
\begin{eqnarray}
2R_1 + R_2 & \leq & \sup_{0 \leq |\rho| \leq 1} \{ \log \left( 1 +
\frac{(1 - |\rho|^2) \mathsf{SNR}}{1 + (1 - |\rho|^2)
\mathsf{INR}}\right) + \log (1 + \mathsf{SNR + INR} +
2|\rho|\sqrt{\mathsf{SNR. INR}} )  \nonumber \\
&& + \log\left(1 + \mathsf{INR} + \frac{\mathsf{SNR} - (1 +
|\rho|^2) \mathsf{INR} + 2 |\rho| \sqrt{\mathsf{SNR.INR} }}{1+
\mathsf{INR}}\right) \}
\end{eqnarray}
which is the statement of the Theorem~\ref{th:sum-capg2}.

\subsection{Gap to Capacity}
\label{gap}
Corresponding to the relevant corner point, we show the gap of the
achievable rate pairs described in Table~\ref{table:rates}. First, we
bound the gap for $\alpha \in [0,1]$ and then for $\alpha \in
(1,\infty)$.

\subsubsection{Corner point $\overline{\mathcal{K}}_{\rm A}$ for $(1000)$ feedback model}\label{para:2r1r2-r1}  It is sufficient to consider only two interference regimes. The
achievable rate pair is described by the corner point
$\mathcal{P}_{\rm A}$ in Table~\ref{table:rates}, which is achievable
without feedback.

\paragraph{$\alpha \in \left[0, \frac{1}{2}\right)$} The gaps of the achievable rate $R_2$ from the outer bound is
\begin{eqnarray}
\overline{C}_2 - R_2 & \leq & \log \left( 1 + \frac{ \mathsf{SNR}}{1 +
  \mathsf{INR}}\right) + \log (1 + \mathsf{SNR + INR} +
2\sqrt{\mathsf{SNR. INR}} )  \nonumber \nonumber \\ && + \log\left(1 +
  \mathsf{INR} + \frac{\mathsf{SNR} -  \mathsf{INR}
    }{1+ \mathsf{INR}}\right) - 2\log(1 +
  \mathsf{SNR}) -
\log\left(\frac{\mathsf{SNR}}{2\mathsf{INR}^2}\right) \leq 3 + \log(3)
\end{eqnarray}
and
\begin{eqnarray}
\overline{C}_1 - R_1 & \leq& \log(1 + \mathsf{SNR}) - [\log(\mathsf{SNR}) - 1 ] \leq \log\left(\frac{1 + \mathsf{SNR}}{\mathsf{SNR}}\right) + 1 = 2.
\end{eqnarray}

\paragraph{$\alpha \in \left[\frac{1}{2},  1\right]$} The achievable
rates and the corresponding distances from the outer bounds are
\begin{eqnarray}
\overline{C}_2 - R_2 & \leq &   \log \left( 1 +
\frac{ \mathsf{SNR}}{1 + \mathsf{INR}}\right) + \log (1 + \mathsf{SNR + INR} +
2\sqrt{\mathsf{SNR. INR}} )  \nonumber \nonumber \\
&& + \log\left(1 + \mathsf{INR} + \frac{\mathsf{SNR} -\mathsf{INR} }{1+
\mathsf{INR}}\right)   - 2\log(1 + \mathsf{SNR})
\leq 3.
\end{eqnarray}
Both the outer and inner bounds for $R_1$ are $\log(1 + \mathsf{SNR})$ and
thus the gap is 0. The point $(\log(1 + \mathsf{SNR}), 0)$ is
trivially achievable.

\subsubsection{Corner point $\overline{\mathcal{K}}_{\rm C}$ for $(1000)$ feedback model}
\label{para:2r1r2-r1r2}
The gap between the corner point $\mathcal{P}_{\rm C}$ in
Table~\ref{table:rates} and the outer bound
$\overline{\mathcal{K}}_{\rm C}$, is computed separately for three
different regimes of interference.

\paragraph{$\alpha \in \left[0,\frac{1}{2}\right)$} The distance of the outer
bound from the achievable rate $R_2$ is
\begin{eqnarray}
\overline{C}_2 - R_2 & \leq & \left[\log \left( 1 +
\frac{\mathsf{SNR}}{1 + \mathsf{INR}}\right) + \log\left(1 + \mathsf{SNR} + \mathsf{INR} + 2\sqrt{\mathsf{SNR}.\mathsf{INR}}\right) - \log\left(1 + \mathsf{INR} +  \frac{\mathsf{SNR} -\mathsf{INR} }{1 + \mathsf{INR}}\right)\right] \nonumber \nonumber \\
&& - \left[\log\left(1 + \frac{\mathsf{SNR}}{2\mathsf{INR}}\right) + \log(\mathsf{INR}) - \log(3)\right] 
\leq 1 + 2\log(3). \label{eq:kc1r2}
\end{eqnarray}
The corresponding gap for the achievable rate $R_1$ is
\begin{eqnarray}
\overline{C}_1 - R_1 & \leq &\log\left(1 + \mathsf{INR} + \frac{\mathsf{SNR - INR}}{\mathsf{INR + 1}}\right) -\log\left( 1 + \frac{\mathsf{SNR}}{2\mathsf{INR}}\right) 
\leq 2. \label{eq:kc1r1}
\end{eqnarray}

\paragraph{${\alpha \in \left[\frac{1}{2},\frac{2}{3}\right)}$} The distance of $R_2$ from the outer bound is bounded as follows
\begin{eqnarray}
\overline{C}_2 - R_2 & \leq &  \log \left( 1 +
\frac{ \mathsf{SNR}}{1 + \mathsf{INR}}\right) + \log (1 + \mathsf{SNR + INR} +
2\sqrt{\mathsf{SNR. INR}} ) - \log\left(1 +
\mathsf{INR} + \frac{\mathsf{SNR - INR }}{\mathsf{INR + 1}}\right) \nonumber \\
 && - \left[\log\left(1 + \frac{\mathsf{SNR}}{\mathsf{2INR}}\right) + \log\left(1 + \frac{\mathsf{SNR}}{\mathsf{INR}}\right) - \log(4)\right] 
\leq 3 + \log(3) \label{eq:kc2r2}
\end{eqnarray}
and the distance of $R_1$ from the outer bound is bounded as
\begin{eqnarray}
\overline{C}_1 - R_1 & = &\left[\log\left(1 + \mathsf{INR} + \frac{\mathsf{SNR - INR }}{\mathsf{INR + 1}}\right)\right] - \left[\log\left(1 + \frac{\mathsf{SNR}}{2\mathsf{INR}}\right) + \log\left(1 + \frac{\mathsf{INR}^2}{\mathsf{SNR}}\right) -\log(4)\right] 
\nonumber \\ & \leq & 3. \label{eq:kc2r1}
\end{eqnarray}

\paragraph{${\alpha \in \left[\frac{2}{3}, 1 \right]}$} The distance of $R_2$ from its corresponding outer bound is
\begin{eqnarray}
\overline{C}_2 - R_2 & \leq & \log \left( 1 + \frac{ \mathsf{SNR}}{1 +
  \mathsf{INR}}\right) + \log (1 + \mathsf{SNR + INR} +
2\sqrt{\mathsf{SNR. INR}} ) - \log\left(1 + \mathsf{INR} +
\frac{\mathsf{SNR - INR }}{\mathsf{INR + 1}}\right) \nonumber \nonumber \\ & & -
\left[\log\left(\frac{\mathsf{SNR}}{\mathsf{INR}}\right) + \log\left(1
  + \frac{\mathsf{SNR}}{2\mathsf{INR}}\right) - \log(1.5)\right] 
 \leq  2 + \log(3). \label{eq:kc3r2}
\end{eqnarray}
and the corresponding gap between the outer bound and $R_1$ is given by
\begin{eqnarray}
\overline{C}_1 - R_1 & \leq & \log\left(1 + \mathsf{INR} + \frac{\mathsf{SNR - INR }}{\mathsf{INR + 1}}\right) - [\log \left(1 + \frac{\mathsf{SNR}}{2\mathsf{INR}}\right) + \log\left(
\frac{\mathsf{INR}^2}{\mathsf{SNR}}\right) - \log(3)] \nonumber \nonumber \\
& \leq & 2 + \log(3). \label{eq:kc3r1}
\end{eqnarray}

\subsubsection{Corner point $\overline{\mathcal{K}}_{\rm D}$ for $(1111)$ feedback model}
\label{para:r2-r1r2_l}
Note that the corner point $\overline{\mathcal{K}}_{\rm D}$ is
identical for all $(1\mathsf{xxx})$ feedback models. This corner point
is within a constant gap from the achievable rate pair described by
the corner point $\mathcal{P}_{\rm D}$ in Table~\ref{table:rates}.
\paragraph{$\alpha \in [0,1]$}
\noindent The gap between the achievable rate $R_2$ from $\overline{C}_2$ is
\begin{eqnarray}
\overline{C}_2 - R_2  \leq  \log(1 + \mathsf{SNR} + \mathsf{INR}) - \log\left( 1 + \frac{\mathsf{SNR}}{2\mathsf{INR}}\right) - \log(\mathsf{INR}) + \log(3) \leq 1 + \log(3). \label{eq:kd1r2}
\end{eqnarray} The corresponding gap of $R_1$ from $\overline{C}_1$ is
\begin{eqnarray}
\overline{C_1} - R_1 & \leq & \left[\log \left( 1 + \frac{
    \mathsf{SNR}}{1 + \mathsf{INR}}\right) + \log \left(\frac{1 +
    \mathsf{SNR + INR} + 2\sqrt{\mathsf{SNR. INR}}}{1 + \mathsf{SNR +
      INR}}\right)\right] - \log\left(1 +
\frac{\mathsf{SNR}}{2\mathsf{INR}} \right) \leq 2. \nonumber \\  \label{eq:kd1r1}
\end{eqnarray}

\subsubsection{Corner point $\overline{\mathcal{K}}_{\rm D}$ of $(0110)$ feedback model}\label{sec:kdfromd}
Note that the corner point $\overline{\mathcal{K}}_{\rm D}$ is
identical for $(0110)$ and $(0010)$ feedback models.
\noindent The gap between the achievable rate $R_2$ from $\overline{C}_2$ is
\begin{eqnarray}
\overline{C}_2 - R_2  \leq  \log(1 + \mathsf{SNR})  - \log\left( 1 + \frac{\mathsf{SNR}}{2\mathsf{INR}}\right) - \log(\mathsf{INR}) + \log(3) \leq 1 + \log(3).
\end{eqnarray} The corresponding gap of $R_1$ from $\overline{C}_1$ is
\begin{eqnarray}
\overline{C_1} - R_1 & \leq & \left[\log \left( 1 + \frac{
    \mathsf{SNR}}{1 + \mathsf{INR}}\right) + \log \left(\frac{1 +
    \mathsf{SNR + INR} + 2\sqrt{\mathsf{SNR. INR}}}{1 + \mathsf{SNR +
      INR}}\right)\right] - \log\left(1 +
\frac{\mathsf{SNR}}{2\mathsf{INR}} \right) \leq 2. \nonumber \\
\end{eqnarray}

\subsubsection{Corner point $\overline{\mathcal{K}}_{\rm B}$ of $(0110)$ feedback model}\label{sec:kbfromb2}
Note that the corner point $\overline{\mathcal{K}}_{\rm B}$ is
identical for $(0110)$ and $(0010)$ feedback models. The gap of the
achievable rate pair described by $\mathcal{P}_{\rm B2}$, in the weak
interference regime, described in Table~\ref{table:rates2} is computed
as follows
\begin{eqnarray}
\overline{C}_2  - R_2 & \leq & \log\left(1 + \mathsf{\frac{SNR}{1 + INR}}\right) + \log(1 + \mathsf{SNR + INR + 2\sqrt{SNR.INR}}) \nonumber \\
& &  - \log\left(1 + {\sf SNR}\right) - \min\{\log\left(\frac{\sf SNR}{\sf 2INR}\right), \log\left(\frac{\sf SNR}{\sf 4INR}\right)\} \leq 3 + \log(3)
\end{eqnarray}
and the corresponding gap of $R_1$ from $\overline{C}_1$ is
\begin{equation}
  \overline{C}_1 - R_1 \leq \log(1 + {\sf SNR}) - \log\left(1 + \frac{\sf SNR}{2}\right)
\leq 1
\end{equation}

Now we list the gap for corner points where $\alpha \in (1, \infty)$.

\subsubsection{Corner point $\overline{\mathcal{K}}_{\rm B}$ of $(1000)$ feedback model}\label{para:r1-r_1r_2}
The achievability is described in Section \ref{subsec:onelink_a_l},
and the achievable rate pair is described by $\mathcal{P}_{\rm B}$ in
Table~\ref{table:rates}.
\paragraph{$\alpha \in (1,2)$} The gap for the achievable rate $R_2$ is
\begin{eqnarray}
\overline{C}_2 - R_2 & \leq & \left[\log\left(1 +
  \frac{\mathsf{SNR}}{1 + \mathsf{INR}}\right) + \log(1 + \mathsf{SNR}
  + \mathsf{INR} + 2\sqrt{\mathsf{SNR}.\mathsf{INR}}) - \log(1 +
  \mathsf{SNR})\right] -
\log\left(1 + \frac{\mathsf{INR}}{\mathsf{SNR}}\right) \nonumber \nonumber \\ & \leq & \log(2)
+ \log\left(\frac{\mathsf{SNR} + \mathsf{SNR}^2 +
  \mathsf{SNR}.\mathsf{INR} +
  2\sqrt{\mathsf{SNR}.\mathsf{INR}}.\mathsf{SNR}}{{\sf SNR}^2 + {\sf SNR} + \mathsf{SNR}.\mathsf{INR}
  + \mathsf{INR}}\right) \leq 1 + \log(3). \label{eq:kb1r2}
\end{eqnarray}
The gap between the achievable rate $R_1$ and the outer bound is
\begin{eqnarray}
\overline{C}_1 - R_1  \leq  \log(1 + \mathsf{SNR}) - \log(\mathsf{SNR})\leq  \log(2)  = 1.\label{eq:kb1r1}
\end{eqnarray}

\paragraph{$\alpha \in [2, \infty)$} The outer bound
  on $R_2$ and its gap from the outer bound is
\begin{eqnarray}
\overline{C}_2 - R_2 & \leq &  \left[\log(1 + \mathsf{INR} + \mathsf{SNR} + 2\sqrt{\mathsf{INR}.\mathsf{SNR}}) + \log\left(1 + \frac{\sf SNR}{\sf INR + 1}\right)- \log(1 + \mathsf{SNR})\right] - \left[\log\left(\frac{\mathsf{INR}}{\mathsf{SNR}}\right) \right] \nonumber \nonumber \\
& \leq & \log\left(\frac{\mathsf{SNR} + \mathsf{INR}.\mathsf{SNR} + \mathsf{SNR}^2 + 2\sqrt{\mathsf{INR}.\mathsf{SNR}}.{\sf SNR}}{\mathsf{INR} + \mathsf{INR}.\mathsf{SNR}}\right) + 1  \leq 
 \log(3) + 1. \label{eq:kb2r2}
\end{eqnarray}
The gap from the outer bound for the achievable rate $R_1$ is
\begin{eqnarray}
\overline{C}_1 - R_1 \leq  \log\left(1 + \mathsf{SNR}\right) - [\log(\mathsf{SNR}) ] \leq  \log\left(1 + \frac{1}{\mathsf{SNR}}\right)  \leq \log(2)  =1. \label{eq:kb2r1}
\end{eqnarray}

\subsubsection{Corner point $\overline{\mathcal{K}}_{\rm D}$ of $(1111)$ feedback model} \label{para:r2-r1r2_h}
The corner point $\overline{\mathcal{K}}_{\rm D}$ is identical for all
$(1\mathsf{xxx})$ feedback models. The achievable rate pair is
described by $\mathcal{P}_{\rm D}$ in Table~\ref{table:rates}.
\paragraph{$\alpha \in (1, \infty)$} The gap between $R_2$ and the outer bound is
\begin{eqnarray}
  \overline{C}_2 - R_2 &\leq & \log(1 + \mathsf{INR} + \mathsf{SNR}) - \log(1 + \mathsf{INR})  \leq  1\label{eq:kd2r2}
\end{eqnarray}
The achievable rate $R_1 = 0$ and the corresponding gap from the outer
bound corner point is
\begin{eqnarray}
\overline{C}_1 - R_1 &\leq & \log\left(1 + \frac{\mathsf{SNR}}{1 + \mathsf{INR}}\right) + \log(1 + \mathsf{INR} + \mathsf{SNR} + 2\sqrt{\mathsf{INR}.\mathsf{SNR}}) - \log(1 + \mathsf{INR} + \mathsf{SNR}) \nonumber \nonumber \\
& \leq & \log(2) + \log\left(1 + \frac{2\sqrt{\mathsf{SNR}.\mathsf{INR}}}{1 + \mathsf{INR} + \mathsf{SNR}}\right)  \leq  \log(2) + \log(2) = 2. \label{eq:kd2r1}
\end{eqnarray}

\subsubsection{Corner point $\overline{\mathcal{K}}_{\rm D}$ of $(0110)$ feedback model}\label{sec:kdnfb}
When $1 < \alpha < 2$, the rate pair described by $\mathcal{P}_{\rm
  D2}$ in Table~\ref{table:rates2} is achievable without any feedback. Its gap from $\overline{\mathcal{K}}_{\rm D}$ is computed as
follows:

\paragraph{$\alpha \in (1,2)$}
\begin{eqnarray}
\overline{C}_1 - R_1 & \leq & \log \left(1 + \frac{\mathsf{SNR}}{ 1 +
  \mathsf{INR}}\right) + \log(1 + \mathsf{SNR + INR} +
2\sqrt{\mathsf{SNR.INR}}) - \log(1 + \mathsf{SNR}) - \log\left(1 +
\frac{\mathsf{INR}}{\mathsf{SNR}} \right)  \nonumber \\
 & \leq & \log(2) + \log(2) + \log \left(\frac{\mathsf{SNR} + \mathsf{SNR}^2
  + \mathsf{INR.SNR} +
  2\mathsf{SNR}.\sqrt{\mathsf{SNR}.\mathsf{INR}}}{\mathsf{2SNR + 2SNR^2
    + 2INR + 2INR.SNR }} \right) \leq 2.
\end{eqnarray}

\begin{eqnarray}
\overline{C}_2 - R_2 & \leq & \log(1 + \mathsf{SNR})  -
\log(\mathsf{SNR}) \leq  1.
\end{eqnarray}



\end{document}